\documentclass[twocolumn,10pt]{article}
\usepackage{authblk}
\usepackage[letterpaper, total={7in, 8.9in}]{geometry}
\usepackage[font=small,labelfont=bf]{caption}

\usepackage{cite} 		  
\usepackage{hyperref}	  
\usepackage{breakurl}
\hypersetup{colorlinks=true,linkcolor=blue,citecolor=blue,breaklinks={true}} 
\usepackage[normalem]{ulem}

\usepackage{graphicx}
\graphicspath{{./FIG/}}
\DeclareGraphicsExtensions{.eps,.png,.jpg,.jpeg,.bmp,.gif,.pdf}

\usepackage{tabularx,booktabs}	
\usepackage{multirow}	        
\usepackage{array}
\newcolumntype{?}{!{\vrule width 2pt}}	
\usepackage{enumitem}           

\usepackage{amsmath}      		
\usepackage{amssymb}	  		
\usepackage{cancel}
\usepackage{bm}		      		
\usepackage{physics}	  		
\usepackage{mathtools}			
\usepackage{dsfont}		  		
\usepackage{soul}				
\usepackage{arydshln}           
\usepackage[version=4]{mhchem}  
\usepackage{multicol}           

\newcommand{\R}{\mathbb{R}}		 
\newcommand{\C}{\mathbb{C}}		 
\newcommand{\transp}{\mathsf{T}} 

\newcommand*{\QEDS}{\hfill\ensuremath{\square}}

\usepackage{amsthm}
\newtheorem{defn}{Definition}
\newtheorem{thm}{Theorem}
\newtheorem{cor}{Corollary}
\newtheorem{prop}{Proposition}
\newtheorem{lem}{Lemma}
\theoremstyle{definition}
\newtheorem{assum}{Assumption}
\newtheorem{rem}{Remark}

\newcommand{\rev}[1]{#1}

\begin{document}

\title{\textbf{\LARGE Identifiability of Differential-Algebraic Systems}}

\author[1]{Arthur N. Montanari$^{*,}$} 
\author[2]{François Lamoline} 
\author[3]{Robert Bereza} 
\author[2,4]{Jorge Gonçalves}

\affil[1]{\normalsize Department of Physics and Astronomy, Northwestern University, Evanston, IL 60208, USA} 
\affil[2]{Luxembourg Centre for Systems Biomedicine, University of Luxembourg, Belvaux  L-4367,  Luxembourg}
\affil[3]{Division of Decision and Control Systems, KTH Royal Institute of Technology, Stockholm SE-100 44, Sweden}
\affil[4]{Department of Plant Sciences, Cambridge University, Cambridge CB2 3EA, United Kingdom}

\twocolumn[
  \begin{@twocolumnfalse}
      \maketitle
Data-driven modeling of dynamical systems often faces numerous data-related challenges. A fundamental requirement is  the existence of a unique set of parameters for a chosen model structure, an issue commonly referred to as identifiability.
Although this problem is well studied for ordinary differential equations (ODEs), few studies have focused on the more general class of systems described by differential-algebraic equations (DAEs).  Examples of DAEs include dynamical systems with algebraic equations representing conservation laws or approximating fast dynamics.
This work introduces a novel identifiability test for models characterized by nonlinear DAEs. Unlike previous approaches, our test only requires prior knowledge of the system equations and does not need nonlinear transformation, index reduction, or numerical integration of the DAEs.
We employed our identifiability analysis across a diverse range of DAE models, illustrating how system identifiability depends on the choices of sensors, experimental conditions, and model structures.
Given the added challenges involved in identifying DAEs when compared to ODEs, we anticipate that our findings will have broad applicability and contribute significantly to the development and validation of data-driven methods for DAEs and other structure-preserving models.

  \vspace{1cm}
  \end{@twocolumnfalse}
]

{\let\thefootnote\relax\footnote{{$^{*}$Corresponding author: {arthur.montanari@northwestern.edu} }}}

\setcounter{footnote}{0}

\vspace{-1cm}
\section{\large \large Introduction}
\label{sec.intro}

Differential-algebraic equations (DAEs) are a powerful tool for the modeling, analysis, and control of constrained dynamical systems \cite{KunkelBook,DuanBook2010}. 
In the most general case, DAEs are represented in implicit form $\bm F(\dot{\bm x},\bm x) = 0$, where $\bm x(t)$ is the $n$-dimensional state of the system and the Jacobian matrix $\partial\bm F(\dot{\bm x},\bm x)/\partial \bm x$ can be singular; in particular, when the Jacobian is nonsingular, DAEs can be reduced to ordinary differential equations (ODEs). Due to its singular form, DAE models require the development of novel mathematical and computational tools to handle the constraints imposed in the state-space exploration of the system \cite{Petzold1982}. This is evident for DAEs represented in semi-explict form, where the model is decomposed as a set of ODEs $\dot{\bm x}_1 = \bm f_1 (\bm x_1,\bm x_2)$ and a set of algebraic equations $\bm f_2 (\bm x_1,\bm x_2) = 0$. The latter equations directly capture constraints imposed in a wide range of systems by conservation laws (e.g., chemical processes \cite{Dudley2019sensitivity} and power grids \cite{Nabavi2016structured}), mechanical constraints (e.g., multibody mechanics and robotics \cite{pogorelov1998differential}), and timescale separations (e.g., singularly perturbed systems like biological regulatory networks \cite{del2012contraction}). 
DAE models directly preserve the structure of the underlying physical system\textemdash a feature that is often lost after reductionist methods are applied to obtain equivalent ODE models \cite{DuanBook2010} (e.g., Kron reduction in power grids \cite{dorfler2012kron}). 
To leverage the structure and constraint preservation in DAEs, many methods in control \cite{chang2001decentralized,Kunkel2001analysis,campbell2016solving,nadeem2023wide} and state estimation \cite{Gerdin2006nonlinear,mandela2010recursive,Nugroho2022observers} have been proposed for a wide variety of applications. Not surprisingly, controllability and observability conditions were also established for linear \cite{Campbell1991duality,Hou1999,DuanBook2010} 
and nonlinear \cite{Terreu1997,Sato2014algebraic} DAE models.

Closely related to the problems of state estimation and observability are the problems of parameter estimation and identifiability in system identification. Given the structure-preserving nature of most DAE models, data-driven methods for system identification of DAEs could be proven very useful for a variety of applications that have been restricted to ODE models thus far; for instance, to uncover structural relations between molecules in regulatory networks \cite{Yuan2011robust,mangan2016inferring,Aalto2020} or interacting components in cyber-physical systems \cite{Yuan19}. Notwithstanding, the identification of DAEs is a severely underdeveloped field, especially when compared to the identification of ODEs \cite{schon2011system,aguirre2019bird,schoukens2019nonlinear}. The majority of methods developed for parameter estimation in DAEs are application-specific (cf. applications in metabolic models \cite{Dudley2019sensitivity} and power grids \cite{Nabavi2016structured}), with only \textit{two} methods available in the literature for general classes of 
linear \cite{gerdin2007parameter} and nonlinear \cite{Abdalmoaty2021,Bereza2022} DAE models. 
However, even for the simple low-dimensional DAE model of a pendulum (which involves only three parameters), the estimation of certain combinations of parameters using the above tools can be unfeasible (as shown in this paper).

The practical limitations in DAE identification lead to the following \textit{identifiability} questions: How to determine which parameters can be estimated from data in DAE models? What data (with respect to measured variable, time-series length, or excitability conditions) should be available for system identification? While there are many analytical \cite{bellman1970structural,Walter1976identifiability,tunali1987new,Gonccalves2008necessary,Villaverde2016structural}, simulation-based \cite{Raue09,stigter2015fast,joubert2021assessing}, and data-driven \cite{evangelou2022parameter} identifiability tests for ODE systems, just a few studies have focused on DAE systems \cite{gerdin2006identification,ben2006identifiability,gerdin2006using,ljung1994global,ashyraliyev2009systems}. Moreover, identifiability tests for DAEs have either involved several nonlinear system transformations and/or index reductions \cite{gerdin2006identification,ljung1994global,chen2013new}, been restricted to linearized models \cite{ben2006identifiability}, or required extensive numerical integration of the system for many initial conditions \cite{gerdin2006using,ashyraliyev2009systems}. As a result, such tests are often time-consuming, cumbersome, or infeasible to computationally implement and test for many real-world applications.


This paper presents an identifiability test for nonlinear DAE systems that only requires prior knowledge of the model structure, measurement functions, and their successive derivatives (Section \ref{sec.idft}). 
The theoretical results are based on a generalization of rank-based conditions for the observability analysis of DAE models \cite{Terreu1997} (reviewed in Section \ref{sec.obsvtest}). As a side contribution, we rigorously show how the observability test for nonlinear DAE models reduces to well-known observability tests for nonlinear ODEs as well as linear DAEs and ODEs\textemdash formalizing the relation between previous work in the literature (Section~\ref{sec.obsvrelation}). Analogously, our identifiability test can also be employed for systems described by linear DAEs and (non)linear ODEs. This allows us to establish some special conditions for the identifiability of linear DAEs based on the structure of the system matrices (Section \ref{sec.linear}).
%
The proposed identifiability analysis is demonstrated in several applications (Section \ref{sec.exmp}). These examples discuss the dependence of the system identifiability on the placement of sensors, excitation regions, and system structure. In particular, we show that, even for linear DAEs with full-state measurement, the identification process can be challenging and reliant on the system structure.

\medskip\noindent
\textbf{Notation.} 
The following terminology and notation are adopted throughout the paper. 
The term \textit{nonsingular systems} is used to refer to dynamical systems described by ODEs [model \eqref{eq.nonsingularsys}] and the term \textit{descriptor systems} to singular systems described by DAEs [implicit \eqref{eq.implicitdescriptor} and semi-explicit \eqref{eq.DAEsemiexplicit} models]. 
%
%
Let $I_n$ denote the identity matrix of size $n$ and $0_{m\times n}$ denote an $m\times n$ null matrix (subscripts are omitted when self-evident). 
$\operatorname{col}(X)$ denotes the column space of $X$.
The vectorization of an $m\times n$ matrix $X$ is performed column-wise, as defined by $\operatorname{vec}(X) = [X_{11} \,\, X_{21} \,\, \ldots \,\, X_{m1} \,\, \ldots \,\, X_{1n} \,\, X_{2n} \,\, \ldots \,\, X_{mn}]^\transp$. The vectorization of multiple matrices is denoted by $\operatorname{vec}(X,\ldots,Y) = [\operatorname{vec}(X)^\transp \,\, \ldots \,\, \operatorname{vec}(Y)^\transp]^\transp$. The Kronecker product is denoted by $X\otimes Y$.

Column vectors are represented by bold letters $\bm x$. The vector $\bm y^{(\nu)}:={\rm d}^\nu \bm y(t)/{\rm d}t^\nu$ denotes the $\nu$th derivative of $\bm y$ with respect to time $t$.
Composition maps are denoted as $\bm f\circ \bm h(\bm x) := \bm f(\bm h(\bm x))$.
%
Let $\mathcal L^\nu_{\bm f} h_j(\bm x)$ be the $\nu$-th Lie derivative of the $j$-th component of $\bm h(\bm x)$ along the vector field $\bm f(\bm x)$.
By definition, $\mathcal L^0_{\bm f} h_j(\bm x)\coloneqq h_j(\bm x)$ and 
\begin{equation*}
    \mathcal L^\nu_{\bm f} h_j(\bm x)\coloneqq\pdv{\mathcal L^{\nu-1}_{\bm f} h_j(\bm x)}{\bm x}\cdot \bm f(\bm x).
\label{eq.liederivatives}
\end{equation*}
\noindent
For compactness, let $\bm{\mathcal L}_{\bm f}^\nu \bm h(\bm x) \coloneqq [\mathcal L_{\bm f}^\nu h_1(\bm x) \, \ldots \, \mathcal L_{\bm f}^\nu h_q(\bm x)]^\transp$.

\section{\large Observability of Descriptor Systems}
\label{sec.obsvdae}


Consider the continuous-time autonomous implicit descriptor system
\begin{subequations}
    \begin{align}
        0 &= \bm F(\bm x,\dot{\bm x}), 
        \label{eq.dynamicsimplicit}
        \\
        \bm y &= \bm h(\bm x),
        \label{eq.measurementimplicit}
    \end{align}
\label{eq.implicitdescriptor}%
\end{subequations}%
\noindent
where $\bm x\in\mathcal X\subseteq\R^n$ is the state vector, $\dot{\bm x}\in\bar{\mathcal X}\subseteq\R^{n}$ is the time derivative of the state vector, and $\bm y\in\R^q$ is the output vector (measurement signals). 
Functions $\bm F:\mathcal X\times\bar{\mathcal X} \mapsto \mathcal F\subseteq \R^n$ and $\bm h:\mathcal X\mapsto \mathcal H\subseteq\R^q$ are considered to be sufficiently smooth. 
Here, we assume that the descriptor system \eqref{eq.implicitdescriptor} can be represented in semi-explicit DAE form (see Refs. \cite{Campbell1995,KunkelBook} for sufficient conditions in which this assumption holds):
\begin{subequations}
    \begin{align}
        \dot{\bm x}_1 &= \bm f_1(\bm x_1,\bm x_2), 
        \label{eq.differentialsys}
        \\
        0 &= \bm f_2(\bm x_1,\bm x_2), \label{eq.algebraicsys} 
        \\
        \bm y &= \bm h(\bm x_1,\bm x_2),
        \label{eq.semiexplicit.output}
    \end{align}
\label{eq.DAEsemiexplicit}%
\end{subequations}
\noindent
where $\bm x_1\in\mathcal X_1\subseteq\R^{n_1}$ and $\bm x_2\in\mathcal X_2\subseteq\R^{n_2}$ are the differential and algebraic variables, respectively ($\bm x := [\bm x_1^\transp \,\,\, \bm x_2^\transp]^\transp$). 

Uniqueness of solutions is considered in the context of initial value problems, meaning that solutions are required to satisfy a consistent initial condition $\bm x(t_0)$. An initial condition $\bm x(t_0)$ is said to be consistent if the associated initial value problem has at least one solution \cite{Pantelides1988}. Let us introduce the set of consistent states  
\begin{equation*}
    \mathbb{L}:=\left\lbrace (t,\bm x, \dot{\bm x}) \Big| \dot{\bm x}_1=\bm f_1(\bm x_1,\bm x_2), 0=\bm f_2(\bm x_1,\bm x_2)  \right\rbrace\subseteq\mathcal X.
\end{equation*}
\noindent
From \cite[Theorem 4.13]{KunkelBook}, we deduce the existence and uniqueness of a local solution to \eqref{eq.differentialsys}--\eqref{eq.algebraicsys}:

\begin{lem} \label{lemma.daesolution}
    Assume that $\mathbb{L}\neq \emptyset$ and that $\bm f_1$ and $\bm f_2$ are $C^1$-smooth functions. Then, for every $\bm x(t_0) = [\bm x_1^\transp(t_0) \,\, \bm x_2^\transp(t_0)]^\transp \in \mathbb{L}$, the semi-explicit DAE form \eqref{eq.DAEsemiexplicit} has a unique local solution satisfying the initial value given by $\bm x(t_0)$.
\end{lem}

Considering only consistent states  $\bm x\in\mathbb L$, let ${\bm\Phi}_T(\bm x(t_0);\bm\theta)$ be the unique solution of system \eqref{eq.DAEsemiexplicit}:
\begin{equation}
    \begin{aligned}
        {\bm\Phi}_T(\bm x(t_0);\bm\theta) &:= \bm x(t_0 + T) 
        \\
        &= \bm x_1(t_0) + \int_{t_0}^{t_0+T} \bm f_1(\bm x(t);\bm\theta){\rm d}t
        \\
        &\quad\,\,\, \textrm{s.t.} \quad 0 = \bm f_2(\bm x_1,\bm x_2;\bm\theta),
        \end{aligned}
    \label{eq.daesolution}
\end{equation}
\noindent
We explicitly denote in ${\bm\Phi}_T(\bm x(t_0);\bm\theta)$ the internal parameters $\bm\theta$ of $\bm f_1$ and $\bm f_2$ (and, by equivalence, of the implicit function $\bm F$), which are relevant for the identifiability analysis. 


We now formalize the local observability property for descriptor systems as a generalization of the local (weak) observability of nonsingular dynamical systems \cite{Hermann1977,VidyasagarBook}.

\begin{defn} \label{def.obsvdae}
	The descriptor system \eqref{eq.implicitdescriptor}, or the pair $\{\bm F,\bm h\}$, is \textit{locally observable at} $\bm x_0$ if there exists a neighborhood $\mathcal U\subseteq\mathbb L$ of $\bm x_0$ such that, for every state $(\bm x'\in\mathcal U) \neq\bm x_0$, $\bm h\circ {\bm\Phi}_T(\bm x_0;\bm\theta)\neq\bm h\circ{{\bm\Phi}}_T(\bm x';\bm\theta)$ for some finite time interval $t\in [t_0,t_0+T]$. 
\end{defn}

Assuming prior knowledge of the system model \eqref{eq.implicitdescriptor}, a descriptor system is locally observable around a state $\bm x_0$ if $\bm x_0$ can be uniquely reconstructed from the measured signal $\bm y(t)$ over a finite trajectory, as defined by the composition map $\bm h\circ{\bm\Phi}_T(\bm x(t_0);\bm\theta)$.  The local weak observability property is related to the notion of \textit{smooth observability}, implying that there exists a smooth map from $\bm y$ and its successive derivatives $\dot{\bm y},\ddot{\bm y}, \ldots, \bm y^{(\nu)}$ to the state vector $\bm x$ if the system is smoothly observable \cite{Campbell1991,Terreu1997}.

\begin{rem}
    The observability property considered here concerns only the reconstruction of consistent states $\bm x\in\mathbb L$. For the reconstruction of inconsistent states that do not satisfy the constraints in solution \eqref{eq.daesolution} (as described by generalized solutions with impulse terms \cite{Chen2022}), we refer to the notion of impulse observability  \cite{Hou1999,DuanBook2010}.
\label{rem.inconsistentstates}
\end{rem}

\subsection{Observability test}
\label{sec.obsvtest}

The local observability of a nonlinear descriptor system can be verified through an algebraic rank condition \cite{Campbell1995,Terreu1997}, analogous to Kalman's famous observability test. To this end, we first define the notion of 1-fullness for a linear set of equations:
\begin{equation}
    \underbrace{
        \begin{bmatrix}
            M_{11} & M_{12} \\ M_{21} & M_{22}
        \end{bmatrix}
        }_{M}
    \underbrace{
        \begin{bmatrix}
            \bm v_1 \\ \bm v_2
        \end{bmatrix}
        }_{\bm v}
        =
    \underbrace{
        \begin{bmatrix}
            \bm b_1 \\ \bm b_2
        \end{bmatrix}
        }_{\bm b},
    \label{eq.1fullproblem}
\end{equation}
\noindent
where $\bm v = [\bm v_1^\transp \,\, \bm v_2^\transp]^\transp\in\R^m$, $\bm v_1\in\R^{m_1}$, and $\bm v_2\in\R^{m_2}$. Let us define $M_1 = [M_{11}^\transp \,\, M_{21}^\transp]^\transp$ and \rev{$M_2 = [M_{12}^\transp \,\, M_{22}^\transp]^\transp$}.

\begin{defn} \textnormal{\cite{Campbell1995}}
    The linear equation $M\bm v = \bm b$ is \textit{1-full with respect to $\bm v_1$} if Eq.~\eqref{eq.1fullproblem} uniquely determines $\bm v_1$ for any consistent vector $\bm b$ and known matrix $M$.
\end{defn}
 
\begin{lem} \label{lemma.1full}
 \textnormal{\cite{Campbell1995}}
    The following statements are equivalent:

    \begin{enumerate}
        \item the linear equation \eqref{eq.1fullproblem} is 1-full with respect to $\bm v_1$;
        \item $\rank(M_1) = m_1$ and $\operatorname{col}(M_1)\cap\operatorname{col}(M_2) = \emptyset$;
        \item $\rank(M) = m_1 + \rank(M_2)$.
    \end{enumerate}
\end{lem}

From Lemma~\ref{lemma.1full}, a local observability condition for nonlinear descriptor systems can be derived  as follows \cite{Terreu1997}.
First, let us subsequently differentiate model~\eqref{eq.implicitdescriptor}:
\begin{equation}
\label{eq.subsequentdifferentiation}%
    \underbrace{
    \begin{aligned}
        0 &= \bm F(\bm x,\dot{\bm x}), \\
        0 &= {\rm d} \bm F(\bm x,\dot{\bm x})/{\rm d}t, \\
        0 &= {\rm d}^2 \bm F(\bm x,\dot{\bm x})/{\rm d}t^2, \\
        &\vdots \\
        0 &= {\rm d}^\mu \bm F(\bm x,\dot{\bm x})/{\rm d}t^\mu,
    \end{aligned}}_{0=\bar{\bm F}_{\mu}(\bm x,\dot{\bm x},\ddot{\bm x},\ldots,\bm x^{(\mu+1)})}
    \quad\quad
    \underbrace{
    \begin{aligned}
    \bm y &= \bm h(\bm x), \\
    \dot{\bm y} &= {\rm d} \bm h(\bm x)/{\rm d}t, \\
    \ddot{\bm y} &= {\rm d}^2 \bm h(\bm x)/{\rm d}t^2, \\
    &\vdots \\
    \bm y^{(\nu)} &= {\rm d}^\nu \bm h(\bm x)/{\rm d}t^\nu.
    \end{aligned}}_{\bar{\bm y} = \bar{\bm H}_{\nu}(\bm x,\dot{\bm x},\ddot{\bm x},\ldots,\bm x^{(\nu)})}
\end{equation}
\noindent
where $\mu\geq 0$ and $\nu\geq 0$ define respectively the number of times $\bm F$ and $\bm h$ are differentiated. 
According to Definition~\ref{def.obsvdae}, the descriptor system \eqref{eq.implicitdescriptor} is locally observable at $\bm x_0$ if there exists a local injective map from the system model, the output signals, and its derivatives up to some order [i.e., Eqs.~\eqref{eq.subsequentdifferentiation}] to the system state $\bm x_0$. Considering the Jacobian matrix of Eqs.~\eqref{eq.subsequentdifferentiation} locally evaluated at $\bm x(t) = \bm x_0$, we have the following local linear map:
\begin{equation}
    \underbrace{
    \begin{bmatrix}
            \dfrac{\partial \bar{\bm F}_{\mu}(\bm x,\bm w)}{\partial \bm x} && \dfrac{\partial\bar{\bm F}_{\mu}(\bm x,\bm w)}{\partial\bm w} \\
            \dfrac{\partial\bar{\bm H}_{\nu}(\bm x,\bm w)}{\partial \bm x} && \dfrac{\partial\bar{\bm H}_{\nu}(\bm x,\bm w)}{\partial\bm w} \\
        \end{bmatrix}
        }_{\mathcal O(\bm x,\bm w)}
        \begin{bmatrix}
            \bm x \\
            \bm w
        \end{bmatrix}
        =
        \begin{bmatrix}
            0_{n\times 1} \\
            \bar{\bm y}
        \end{bmatrix}
\label{eq.daeobsvmatrix}
\end{equation}
\noindent
where $\bm w = [\dot{\bm x}^\transp \,\, \ddot{\bm x}^\transp \,\, \ldots \,\, \bm x^{(\sigma) \, \transp}]^\transp$, with order $\sigma = \max\{\mu+1,\nu\}$, and $\mathcal O(\bm x,\bm w)$ is the observability matrix. Since $\bm F$ and $\bm h$ are smooth functions, following the Constant Rank Theorem, an algebraic observability condition can thus be derived for the unique reconstruction of $\bm x_0$ by testing under which conditions the linear set of equations \eqref{eq.daeobsvmatrix} is 1-full with respect to $
\bm x$:

\begin{thm} \label{theor.nonlinearobsv} {\normalfont \textbf{(Observability of nonlinear DAEs)}}~\textnormal{\cite{Terreu1997}} 
	The descriptor system \eqref{eq.implicitdescriptor}, or the pair $\{\bm F, \bm h\}$, is locally observable at $\bm x_0\in\mathbb L$ if
    \begin{equation}
        \operatorname{rank}
        (\mathcal O(\bm x,\bm w))
        = n + \operatorname{rank}
        \begin{bmatrix}
            \dfrac{\partial\bar{\bm F}_{\mu}(\bm x,\bm w)}{\partial\bm w} \\
            \dfrac{\partial\bar{\bm H}_{\nu}(\bm x,\bm w)}{\partial\bm w} \\
        \end{bmatrix}.
    \label{eq.daeobsv}
    \end{equation}
	\noindent
	holds at $\bm x = \bm x_0$ for some $\mu$ and $\nu$.

\end{thm}
%


\begin{rem}
\label{rem.differentiation}
    The rank condition~\eqref{eq.daeobsv} depends on the number of times that functions $\bm F$ and $\bm h$ are differentiated in order to construct the observability matrix $\mathcal O(\bm x,\bm w)$. In practice, we iteratively increment $\mu$ and $\nu$ until the image space of $\mathcal O$ stops growing for some $\sigma$, i.e., until $\rank(\mathcal O(\bm x,\dot{\bm x},\ldots,\bm x^{(\sigma)})) = \rank(\mathcal O(\bm x,\dot{\bm x},\ldots,\bm x^{(\sigma+1)}))$.
\end{rem}

\begin{rem}
    Although here we focus on autonomous systems \eqref{eq.implicitdescriptor}, the results can be directly generalized to control systems $\bm F(\bm x,\dot{\bm x}) = B\bm u$ or $\bm F(\bm x,\dot{\bm x},\bm u) = 0$, where $\bm u$ is the input vector (control signal).
\label{rem.feedbacksys}
\end{rem}


\subsection{Relation to special cases}
\label{sec.obsvrelation}

The observability test for nonlinear descriptor systems proposed in Ref. \cite{Terreu1997} was developed independently from other well-known observability tests for linear descriptor systems 
\cite{Hou1999} and (linear and nonlinear) nonsingular systems \cite{Kalman1959,Hermann1977}. As such, a clear relationship between these tests is not yet available in the literature. As a first contribution of our paper, we derive and formalize the relationship between the aforementioned observability tests, showing that Theorem~\ref{theor.nonlinearobsv} is  indeed a rigorous generalization of observability properties presented in previous work. These relationships are summarized in Fig.~\ref{fig.generalization}.


\medskip\noindent
\textbf{Nonsingular systems.}
The differential index of a system of DAEs is the number of times the DAEs have to be differentiated in order to be expressed as an equivalent system of ODEs, also known as nonsingular systems. Thus, the well-studied class of nonsingular systems described by
\begin{subequations}
    \begin{align}
        \dot{\bm x} &= \bm f(\bm x), \\
        \bm y &= \bm h(\bm x),
    \end{align}
\label{eq.nonsingularsys}%
\end{subequations}%
\noindent
can be seen as a special form of descriptor systems \eqref{eq.implicitdescriptor} with index 0, where $\bm F(\bm x,\dot{\bm x}) = \bm f(\bm x) - \dot{\bm x}$. 
%
%
%
An observability condition for nonsingular systems can thus be derived from Theorem~\ref{theor.nonlinearobsv}:

\begin{prop} {\normalfont \textbf{(Observability of nonlinear ODEs)}}
\label{corol.obsvnonsingular}
    The nonsingular system \eqref{eq.nonsingularsys}, or the pair $\{\bm f,\bm h\}$ is locally observable at $\bm x_0$ if the following equivalent conditions are satisfied:
    \begin{enumerate}
        \item $\mathcal O(\bm x,\bm w)$ is nonsingular at $\bm x=\bm x_0$ for some $\mu$ and $\nu$; 
        
        \item the Jacobian matrix
        \begin{equation*}
            \bm \Psi(\bm x) = \pdv{}{\bm x}
            \begin{bmatrix}
            \mathcal L_{\bm f}^0\bm h(\bm x) \\ \mathcal L_{\bm f}^1\bm h(\bm x) \\ \vdots \\  \mathcal L_{\bm f}^\nu\bm h(\bm x)
        \end{bmatrix} \label{eq.obsvmatrix.liederivative}
        \end{equation*}
        is full-column rank at $\bm x=\bm x_0$ for some $\nu$.   
    \end{enumerate}
\end{prop}

\noindent
    \textit{Proof sketch.} 
    Condition 2 is a well-known condition in the literature for the observability of nonlinear nonsingular systems; proofs can be found in Refs. \cite{Hermann1977,VidyasagarBook,Montanari2020}. Here, we show that condition 1 follows directly from Theorem~\ref{theor.nonlinearobsv} for $\bm F(\bm x,\dot{\bm x}) = \bm f(\bm x) - \dot{\bm x}$. Furthermore, we prove that conditions 1 and 2 are equivalent. See Appendix \ref{app.proofs} for details.  \QEDS

\begin{figure}
    \centering
    \includegraphics[width=\columnwidth]{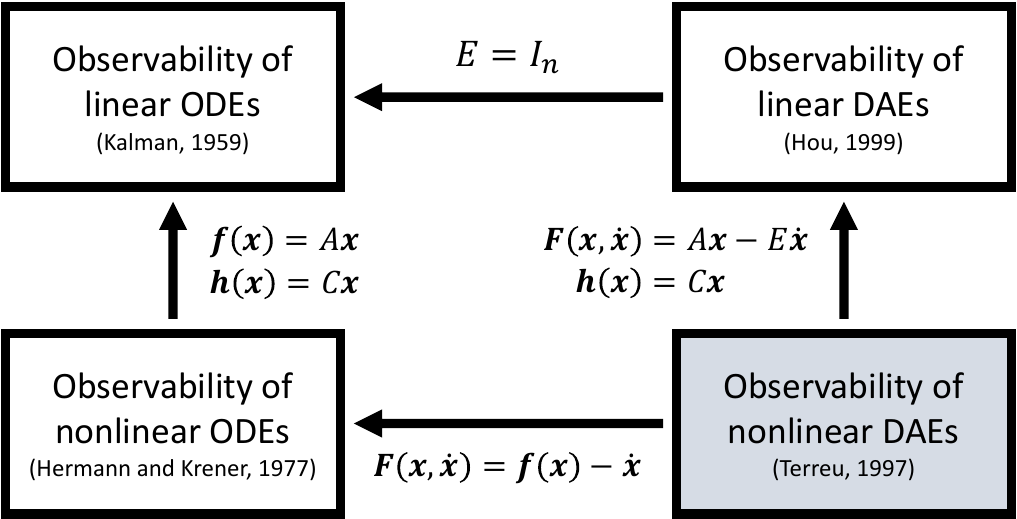}
    \caption{Observability of descriptor systems and its special cases.}
    \label{fig.generalization}
\end{figure}

    Another important class of descriptor systems are represented by semi-explicit DAEs \eqref{eq.DAEsemiexplicit} with index 1~\cite{Winkler2004}, meaning that differentiating the constraints \eqref{eq.algebraicsys} once leads to an ODE.
    Since $\bm f_2$ is smooth, an index-1 DAE implies that Eq. \eqref{eq.algebraicsys} is locally solvable for the algebraic variable $\bm x_2$. Following the Implicit Function Theorem, if $\rank(\partial \bm f_2(\bm x_1,\bm x_2)/\partial \bm x_2) = n_2$, then there exists a locally unique solution $\bm x_2 = \bm F_2(\bm x_1)$ which inserted into Eqs. \eqref{eq.differentialsys} and \eqref{eq.semiexplicit.output} yields
    \begin{subequations} \label{eq.index1DAE}
    \begin{align}
        \dot{\bm x}_1 &= \bm f_1(\bm x_1,\bm F_2(\bm x_1)) := \bm F_1(\bm x_1),
        \\
        \bm y &= \bm h(\bm x_1,\bm F_2(\bm x_1)) \,\, := \bm H(\bm x_1).
    \end{align}
    \end{subequations}
    \noindent
    The local observability of an index-1 descriptor system can, therefore, be verified using Proposition~\ref{corol.obsvnonsingular} for the pair $\{\bm F_1,\bm H\}$ at any state $\bm x_0\in\mathbb L$. From Definition~\ref{def.obsvdae}, if $\{\bm F_1,\bm H\}$ is locally observable, then $\bm x_1$ can be locally inferred from the measurement signal $\bm y$ and hence $\bm x_2$ can be obtained from its local solution $\bm F_2(\bm x_1)$. However, applying Proposition \ref{corol.obsvnonsingular} requires the construction of the local map $\bm F_2$, which can be significantly more cumbersome to handle compared to the direct test on the original descriptor system's equations proposed in Theorem~\ref{theor.nonlinearobsv}. 

\medskip\noindent
\textbf{Linear systems.}
If the pair $\{\bm F,\bm h\}$ is given by linear functions, then system \eqref{eq.implicitdescriptor} can be represented as
\begin{subequations}
    \begin{align}
        E\dot{\bm x} &= A\bm x + B\bm u, 
        \\
        \bm y &= C\bm x,
    \end{align}
\label{eq.lineardescriptor}%
\end{subequations}%
\noindent
where $A\in\R^{n\times n}$, $B\in\R^{n\times p}$, and $C\in\R^{q\times n}$. The matrix $E\in\R^{n\times n}$ is possibly singular and we assume that the control signal $\bm u = 0$ without loss of generality. Since we restrict the observability property to the reconstruction of consistent states $\bm x_0\in\mathbb L$ (Remark~\ref{rem.inconsistentstates}), Definition~\ref{def.obsvdae} reduces to the notion of R-observability for linear descriptor systems, which is defined globally as follows \cite{Hou1999,DuanBook2010}:

\begin{defn} \label{def.linearobsv}
 \textnormal{\cite{DuanBook2010}}
    The linear descriptor system \eqref{eq.lineardescriptor} is R-observable if, for any arbitrary $\bm x(t_0)\in\mathbb L$, $\bm x(t_0)$ can be uniquely determined from the output signal $\bm y(t)$ for some finite time interval $t\in [t_0,t_0+T]$.
\end{defn}

We now show that the R-observability condition follows as a special case of Theorem~\ref{theor.nonlinearobsv}:

\begin{prop} \label{corol.lineardescriptor} {\normalfont \textbf{(Observability of linear DAEs)}}
The following statements are equivalent:

\setlength{\abovedisplayskip}{2pt}
\setlength{\belowdisplayskip}{2pt}

\begin{enumerate}
    \item the linear descriptor system \eqref{eq.lineardescriptor} is R-observable;
    
    \item $\rank(O) = n + \rank\begin{bmatrix} O_{12} \\ O_{22} \end{bmatrix}$, where the  $n(n+q)\times n(n+1)$ observability matrix is defined by
    \begin{equation*}
        \hspace{-21pt}
        O = 
        \left[\begin{array}{c|c}
            O_{11} & O_{12} \\ \hline O_{21} & O_{22} 
        \end{array}\right]
        =
        \left[\begin{array}{c|ccccc}
            A & -E & 0 & \ldots & 0 & 0 \\
            0 & A & -E & \ldots & 0 & 0\\
            \vdots & \vdots & \vdots & \ddots  & \vdots & \vdots\\
            0 & 0 & 0 &  \ldots & A & -E \\
            \hline
            C & 0 & 0 & \ldots & 0 & 0 \\
            0 & C & 0 & \ldots & 0 & 0\\
            \vdots & \vdots & \vdots & \ddots  & \vdots & \vdots\\
            0 & 0 & 0 & \ldots & C & 0
        \end{array}\right];
        \end{equation*}
        \item $\rank\begin{bmatrix} \lambda E-A \\ C \end{bmatrix} = n$, $\forall\lambda\in\C$. 
    \end{enumerate}
\end{prop}

\setlength{\abovedisplayskip}{7pt}
\setlength{\belowdisplayskip}{7pt}

\noindent
    \textit{Proof sketch.} 
    Statement 3 is a well-known sufficient and necessary condition in the literature for the R-observability of linear descriptor systems; a complete proof can be found in \cite[Sec. 4.3 and 4.5]{DuanBook2010}. Here, we show that statement 2 is sufficient for R-observability by applying Theorem~\ref{theor.nonlinearobsv} to linear functions. To prove the necessity of statement 2, we show that statements 2 and 3 are equivalent. See Appendix \ref{app.proofs} for details.  \QEDS 

If $E$ is nonsingular, then $\mathbb L = \R^n$ (all states are consistent) and, hence, the notion of R-observability (Definition~\ref{def.linearobsv}) reduces to the well-known observability property for linear nonsingular systems \cite{Kalman1959,Chi-TsongChen1999}. It follows directly from Proposition~\ref{corol.lineardescriptor} that:

\begin{cor} {\normalfont \textbf{(Observability of linear ODEs)}}
If $E$ is nonsingular, then the following statements are equivalent:

\begin{enumerate}[label={\arabic*$'$)}]
    \item the linear descriptor system \eqref{eq.lineardescriptor} is observable;
    
    \item $\rank (O')=n$, where $A' = E^{-1} A$ and
    \begin{equation*}
        O' = \begin{bmatrix} C^\transp & (CA')^\transp & \ldots & (C(A')^{n-1})^\transp \end{bmatrix}^\transp
    \end{equation*}
    is the Kalman's observability matrix;
    
    \item $ \rank\begin{bmatrix} \lambda I-A' \\ C \end{bmatrix} = n, \quad \forall\lambda\in\C$.
\end{enumerate}
\end{cor}

\begin{proof}
If $E$ is nonsingular, then the descriptor model \eqref{eq.lineardescriptor} can be expressed as a linear system of ODEs: $\dot{\bm x} = A'\bm x + B'\bm u$, where $A'=E^{-1}A$ and $B' = E^{-1}B$. 
It thus follows that statement 3 of Proposition~\ref{corol.lineardescriptor} reduces to statement 3$'$ (also known as the Popov-Belevitch-Hautus test).
Likewise, for a nonsingular $E$, statement 2 of Proposition~\ref{corol.lineardescriptor} holds if and only if $O$ is nonsingular (since $O_{12}$ is full rank). Moreover, $O$ is nonsingular if and only if $O_{21} - O_{22}O_{12}^{-1}O_{11} = O'$ is invertible, leading to statement 2$'$.  
\end{proof}


\section{\large Identifiability of Descriptor Systems}
\label{sec.idft}

We now extend the observability condition presented in Theorem \ref{theor.nonlinearobsv} to establish an algebraic identifiability condition for nonlinear descriptor systems.
Let us consider again the autonomous descriptor system
\begin{subequations}
    \begin{align}
        0 &= \bm F(\bm x,\dot{\bm x};\bm \theta), 
        \label{eq.dynamicsimplicit.parameter}
        \\
        \bm y &= \bm h(\bm x;\bm\theta),
        \label{eq.measurementimplicit.parameter}
    \end{align}
\label{eq.implicitdescriptor.parameter}%
\end{subequations}%
\noindent
where we explicitly denote $\bm\theta\in\R^p$ as the vector of unknown parameters sought to be identified. The notion of \textit{parameter identifiability} establishes the conditions under which, given a known model structure \eqref{eq.implicitdescriptor.parameter}, the unknown parameters $\bm\theta$ can be uniquely determined from a measurement signal $\bm y(t)$ recorded over $t\in[t_0,t_0+T]$. This is formally defined as follows.

\begin{defn}
    The descriptor system \eqref{eq.implicitdescriptor.parameter}, or the pair $\{\bm F,\bm h\}$, is \textit{locally identifiable at $\bm x_0\in\mathbb L$ with respect to parameter $\bm\theta_0$} if there exists a neighborhood $\mathcal U\subseteq\R^p$ of $\bm\theta_0$ such that, for every parameter $(\bm\theta'\in\mathcal U)\neq \bm\theta_0$, $\bm h\circ{\bm\Phi}_T(\bm x_0;\bm\theta_0) \neq \bm h\circ{\bm\Phi}_T(\bm x_0;\bm\theta')$ for some finite time interval $t\in [t_0,t_0+T]$. 
\label{def.identifiability}
\end{defn}

Comparing Definitions~\ref{def.obsvdae} and \ref{def.identifiability}, it becomes evident that the local identifiability problem can be treated as an observability problem if the parameters are represented as state variables of the descriptor model\textemdash a framework previously explored for nonsingular systems \cite{tunali1987new,Villaverde2016structural}.
To this end, we first augment the descriptor system \eqref{eq.implicitdescriptor.parameter} with additional state variables representing the parameters $\bm\theta$ sought to be identified (assuming constant dynamics $\dot{\bm\theta} = 0$ so that $\bm\theta(t)=\bm\theta$, for all $t\geq 0$). This yields the following augmented implicit model:
\begin{subequations}
\label{eq.implicitdescriptor.augmented}%
    \begin{align}
        0 &= \bm F(\bm x,\dot{\bm x},\bm \theta),
        \\
        0 &= \dot{\bm\theta},
        \\
        \bm y &= \bm h(\bm x,\bm\theta).
    \end{align}
\end{subequations}%
%
%
%
An algebraic (local) identifiability condition can now be derived for the augmented model \eqref{eq.implicitdescriptor.augmented} based on the 1-fullness property of linear systems (Lemma~\ref{lemma.1full}), analogously to the observability test presented in Theorem~\ref{theor.nonlinearobsv}.

\begin{thm} \label{theor.identifiability} {\normalfont \textbf{(Identifiability of nonlinear DAEs)}} Let the identifiability matrix be defined as
\begin{equation} \label{eq.identifmatrix}
        \mathcal I(\bm\theta, \bm z) = 
        \begin{bmatrix}
            \dfrac{\partial \bar{\bm F}_{\mu}(\bm\theta,\bm z)}{\partial \bm \theta} 
            &&
            \dfrac{\partial\bar{\bm F}_{\mu}(\bm\theta,\bm z)}{\partial\bm z} 
            \\
            \dfrac{\partial\bar{\bm H}_{\nu}(\bm\theta,\bm z)}{\partial \bm \theta} 
            && 
            \dfrac{\partial\bar{\bm H}_{\nu}(\bm\theta,\bm z)}{\partial\bm z}
        \end{bmatrix},
\end{equation}
\noindent
where $\bm z = [\bm x^\transp \,\, \dot{\bm x}^\transp \,\, \ldots \,\, \bm x^{(\sigma) \, \transp}]^\transp$ and $\sigma = \max\{\mu+1,\nu\}$.
The descriptor system \eqref{eq.implicitdescriptor.augmented}, or the pair $\{\bm F,\bm h\}$, is locally identifiable at $\bm x_0\in\mathbb L$ with respect to $\bm\theta_0$ if
    \begin{equation}
        \rank(\mathcal I(\bm\theta,\bm z)) = p + \rank
        \begin{bmatrix}
            \dfrac{\partial\bar{\bm F}_{\mu}(\bm\theta,\bm z)}{\partial\bm z} 
            \\
            \dfrac{\partial\bar{\bm H}_{\nu}(\bm\theta,\bm z)}{\partial\bm z}
        \end{bmatrix}
    \label{eq.identifiabilitytest}
    \end{equation}
    \noindent
    holds at $(\bm x,\bm\theta) = (\bm x_0,\bm\theta_0)$ for some $\mu$ and $\nu$.
\end{thm}

\begin{proof}
    See Appendix \ref{app.proofs}.  
\end{proof}

\vspace{-0.1cm}

\begin{rem}
    The rank condition \eqref{eq.identifiabilitytest} is state dependent, meaning that the system may be identifiable (or unidentifiable) only in a subregion of the parameter space and the set of consistent states (i.e., in a subset of $\R^p\times\mathbb L$). This is illustrated in detail in Section~\ref{example.chemical}.
\end{rem}

\vspace{-0.1cm}

\begin{rem}
\label{rem.almostnecessary}
    The observability condition \eqref{eq.daeobsv} is sufficient and ``almost'' necessary for nonlinear descriptor systems. From \cite[Corollary 112]{Vidya93}, the observability condition is necessary on a dense subset of $\mathcal{X}$, which implies that it is also a necessary condition when restricted to $\mathbb L\subseteq\mathcal X$ if $\mathbb L$ is dense. Thus, not satisfying condition \eqref{eq.daeobsv} is a strong indication on the unobservability of a system. The same argument can be extended to the identifiability test \eqref{eq.identifiabilitytest}.
\end{rem}



\vspace{-0.1cm}

Theorems~\ref{theor.nonlinearobsv} and \ref{theor.identifiability} provide a clear advantage for the observability and identifiability analysis of nonlinear descriptor systems: they establish algebraic conditions, which can be directly checked on the original system's equations \eqref{eq.implicitdescriptor.parameter}. This provides a practical and simple test that can guide experimental design for system identification, including optimal sensor placement (which variables to measure) and excitation conditions (which regions of the state space to explore for data collection in experiments). We explore these advantages in system analysis with examples presented in Section \ref{sec.exmp}.

\vspace{-0.1cm}

An important underlying assumption in Theorems \ref{theor.nonlinearobsv} and \ref{theor.identifiability} is that the observability and identifiability analyses are evaluated locally in the neighborhood of a consistent state (Remark \ref{rem.inconsistentstates}). The definition of a consistent state in this work (set $\mathbb L$) requires an implicit model \eqref{eq.implicitdescriptor} to have an equivalent representation in semi-explicit form \eqref{eq.DAEsemiexplicit}. This is often not a problem since DAEs in semi-explicit form arise naturally in the modeling of many technological, biochemical, and environmental systems. However, if only an implicit model is available, it may not be easy or even possible to derive its corresponding semi-explicit representation.
Despite this challenge, if a consistent state can be determined (e.g., numerically), then the conditions in Theorems~\ref{theor.nonlinearobsv} and \ref{theor.identifiability} can be applied directly to the system's equations in implicit form, not requiring conversion to semi-explicit form.

\subsection{Linear descriptor systems}
\label{sec.linear}

Given that the observability test (Theorem \ref{theor.nonlinearobsv}) is applicable both to (non)linear DAEs and ODEs, as proven in Section \ref{sec.obsvrelation}, it follows that the identifiability test (Theorem \ref{theor.identifiability}) is also applicable to all these classes of systems. In what follows, we explore the structure of linear descriptor systems to determine sufficient conditions in which the identifiability condition \eqref{eq.identifiabilitytest} is always locally satisfied, ensuring system identifiability.

\vspace{-0.1cm}

Consider the linear descriptor system \eqref{eq.lineardescriptor}. In this scenario, the parameters $\bm\theta\in\R^p$ sought to be identified are typically a subset of the matrix elements $A_{ij}$ (i.e., $\theta_k = A_{ij}$ if $A_{ij}$ is an unknown parameter, for $k=1,\ldots,p$ and $p\leq n^2$). To this end, we model the parameters as state variables with constant dynamics [as in system \eqref{eq.implicitdescriptor.augmented}], yielding
\begin{subequations}
\label{eq.lineardescriptor.augmented}%
    \begin{align}
        E\dot{\bm x} &= A(\bm\theta)\bm x,
        \\
        \dot{\bm\theta} &= 0,
        \\
        \bm y &= C\bm x,
    \end{align}
\end{subequations}%
\noindent
where we explicitly denote the dependence of $A$ on the parameters $\bm\theta$.
The identifiability analysis of system \eqref{eq.lineardescriptor.augmented} can thus be conducted by directly applying Theorem \ref{theor.identifiability} on the functions $\bm F(\bm x,\dot{\bm x}) = A\bm x - E{\dot{\bm x}}$ and $\bm h(\bm x) = C\bm x$. 
Since parameters $\bm\theta$ are treated as state variables, the identifiability of linear descriptor systems is typically a \textit{nonlinear} problem due to the matrix-vector multiplication $A(\bm\theta)\bm x$. Thus, Theorem \ref{theor.identifiability} still only establishes a sufficient condition for the identifiability of linear descriptor systems. 
Nonetheless, the particular structure of the descriptor system \eqref{eq.lineardescriptor.augmented} provides a concise form for the identifiability matrix \eqref{eq.identifmatrix}, which can be described by
\begin{align}
    \label{eq.identifmatrix.linear}
        \mathcal I(\bm\theta, \bm z) = 
        \left[\begin{array}{c|c}
            \mathcal I_{11}(\bm\theta,\bm z)
            &
            \Delta_{1}\otimes A - \Delta_2\otimes E
            \\ \hline
            0_{q(\nu+1)\times p}
            & 
            \Delta_{1}\otimes C
        \end{array}\right],
\end{align}
\noindent
where $\Delta_{1} = [I_{\sigma} \,\, 0_{\sigma n\times n}]$, $\Delta_2 = [0_{\sigma n\times n} \,\, I_{\sigma}]$, and we assumed that $\mu = \nu$. The matrix block $\mathcal I_{11}$ is defined as
\begin{equation*}
    \mathcal I_{11}(\bm\theta,\bm z) := \dfrac{\partial}{\partial \bm\theta} \begin{bmatrix}
                A(\bm\theta)\bm x \\
                A(\bm\theta)\dot{\bm x} \\
                \vdots \\
                A(\bm\theta)\bm x^{(\mu)}
            \end{bmatrix}.
\end{equation*}
By applying Theorem \ref{theor.identifiability}, we directly derive the following identifiability test for linear descriptor systems:
\begin{cor}
    \label{cor.identifiability.linear} {\normalfont \textbf{(Identifiability of linear DAEs)}} Consider the identifiability matrix $\mathcal I(\bm\theta,\bm z)$ defined in Eq. \eqref{eq.identifmatrix}. The descriptor system \eqref{eq.lineardescriptor.augmented} is locally identifiable at $\bm x_0\in\mathbb L$ with respect to $\bm\theta_0$ if
    \begin{equation} \label{eq.identcond.linear}
        \rank(\mathcal I(\bm\theta,\bm z)) = p + \rank
        \begin{bmatrix}
            \Delta_{1}\otimes A - \Delta_2\otimes E
            \\
            \Delta_{1}\otimes C
        \end{bmatrix}
    \end{equation}
    \noindent
    holds at $(\bm x,\bm\theta) = (\bm x_0,\bm\theta_0)$ for some $\mu = \nu$.
\end{cor}

\bigskip
\begin{rem}
    For applications in which \textit{all} elements $A_{ij}$ are sought to be identified [i.e., $\bm\theta = \operatorname{vec}(A)$], it follows that
\begin{equation} \label{eq.Athetax}
    \mathcal I_{11}(\bm\theta,\bm z)
    = 
    \begin{bmatrix}
        \bm x^\transp \otimes I_n \\
        \dot{\bm x}^\transp \otimes I_n \\
        \vdots \\
        {\bm x}^{(\mu)^\transp} \otimes I_n \\
    \end{bmatrix} 
    =
    \begin{bmatrix}
        \bm x^\transp \\ \dot{\bm x}^\transp \\ \vdots \\ (\bm x^{(\mu)})^\transp
    \end{bmatrix} \otimes I_n.
\end{equation}
\end{rem}

We now leverage the structure of the identifiability matrix \eqref{eq.identifmatrix.linear} to derive some special, easier-to-check identifiability conditions for linear descriptor systems under the assumption that full-state measurement is available (i.e., $C=I_n$).

\begin{cor} \label{cor.idflinearfullmeas}
{\normalfont \textbf{(Identifiability of linear DAEs with full-state measurement)}}
    Suppose $C = I_n$. The descriptor system \eqref{eq.lineardescriptor.augmented} is locally identifiable at $\bm x_0\in\mathbb L$ with respect to $\bm\theta_0$ if
    \begin{equation} \label{eq.rankAtheta}
        \rank\left(\mathcal I_{11}(\bm\theta,\bm z)\right) = p
    \end{equation}
    \noindent
    holds at $(\bm x,\bm\theta) = (\bm x_0,\bm\theta_0)$ for some $\mu$.
\end{cor}

\begin{proof}
    Recall that the identifiability matrices \eqref{eq.identifmatrix} and \eqref{eq.identifmatrix.linear} are equal for the linear functions $\bm F(\bm x,\dot{\bm x}) = A\bm x - E\dot{\bm x}$ and $\bm h(\bm x) = C\bm x$. If $C=I_n$, then
    \begin{equation*}
        \operatorname{col}\left(\begin{bmatrix}
            \mathcal I_{11}(\bm\theta,\bm z) 
            \\
            0_{q(\nu+1)\times p}
        \end{bmatrix}\right)
        \cap
        \operatorname{col}\left(\begin{bmatrix}
            \Delta_{1}\otimes A - \Delta_2\otimes E
            \\
            \Delta_{1}\otimes C
        \end{bmatrix}\right)
         = \emptyset
    \end{equation*}
\noindent
holds by construction. Hence, due to the equivalence of statements 2 and 3 in Lemma \ref{lemma.1full}, if condition \eqref{eq.rankAtheta} is satisfied, then condition \eqref{eq.identifiabilitytest} is also satisfied.  
\end{proof}

Corollary \ref{cor.idflinearfullmeas} highlights that full-state measurement is not sufficient for the identifiability of general linear descriptor systems. The identifiability will thus depend on the structure of the system (encoded by matrices $A$ and $E$) as well as the set of parameters to be identified. In particular, if the system is represented in semi-explicit form
\begin{equation}
\label{eq.lineardescriptor.semiexplcit}%
        \begin{bmatrix}
            I_{n_1} & 0 \\ 0 & 0
        \end{bmatrix}
        \begin{bmatrix}
            \dot{\bm x}_1 \\ \dot{\bm x}_2
        \end{bmatrix}
        = 
        \begin{bmatrix}
            A_{11} & A_{12} \\ A_{21} & A_{22}
        \end{bmatrix}
        \begin{bmatrix}
            {\bm x}_1 \\ {\bm x}_2
        \end{bmatrix}
\end{equation}
\noindent
further sufficient conditions can be derived depending on the set of parameters sought to be identified. Eq. \eqref{eq.lineardescriptor.semiexplcit} is a special form of the DAE system \eqref{eq.DAEsemiexplicit} for linear functions, where $\bm x_1\in\R^{n_1}$ and $\bm x_2\in\R^{n_2}$ are the differential and algebraic variables, respectively. Matrices $A_{12}$ and $A_{21}$ describe interactions between differential and algebraic variables. 
Depending on the application, certain submatrices may be known a priori and hence the set of unknown parameters can be restricted to the remaining unknown matrices, e.g., $\bm\theta = \operatorname{vec}(A_{11},A_{21})$. The following result shows that some combinations of parameters are identifiable if full-state measurement is available.

\begin{assum} \label{assump.index1lineardae}
The linear DAE \eqref{eq.lineardescriptor.semiexplcit} is index 1, that is, $A_{22}$ is nonsingular. This implies that $\bm x_2 = - A_{22}^{-1}A_{21}\bm x_1$ and hence the DAEs can be reduced to the ODEs $\dot{\bm x}_1 = A_{\rm c}\bm x_1$, where $A_{\rm c} = A_{11} - A_{12}A_{22}^{-1}A_{21}$.
\end{assum}

\begin{cor} \label{cor.idftlinear.semiexplict} {\normalfont \textbf{(Identifiability of semi-explicit linear DAEs)}}
    Suppose $C = I_n$ and Assumption \ref{assump.index1lineardae} holds. If $A_{21}$ is full rank, then the descriptor system \eqref{eq.lineardescriptor.semiexplcit} is locally identifiable at $\bm x_0\neq 0$ with respect to the set of parameters $\bm\theta_0 = \operatorname{vec}(A_{11},A_{22})$ or $\bm\theta_0 = \operatorname{vec}(A_{12},A_{21})$.
\end{cor}

\begin{proof}
    See Appendix \ref{app.proofs}.  
\end{proof}

Other parameter combinations like $\bm\theta = \operatorname{vec}(A_{11},A_{12})$ may also be identifiable in the case of full-state measurement depending on the sparsity of the matrix $A$. These scenarios are further explored numerically in Section \ref{example.linear}.

The identifiability analysis of linear descriptor systems can also be applied to linear nonsingular systems (i.e., when $E=I_n$). For such cases, the following corollary shows that\textemdash contrariwise to descriptor systems\textemdash full-state measurement is sufficient for the complete identification of linear ODE systems (which is a well-known condition in the literature \cite{Delforge, Walter1976identifiability}).

\begin{cor} \label{cor.idftlinear.ode} {\normalfont \textbf{(Identifiability of linear ODEs)}}
    Suppose $C = I_n$ and $E = I_n$. The descriptor system \eqref{eq.lineardescriptor.augmented} is locally identifiable at $\bm x_0\neq 0$ with respect to $\bm\theta = \operatorname{vec}(A)$.

\end{cor}

\begin{proof}
    See Appendix \ref{app.proofs}.  
\end{proof}

\begin{figure*}
    \centering
    \includegraphics[width=0.99\linewidth]{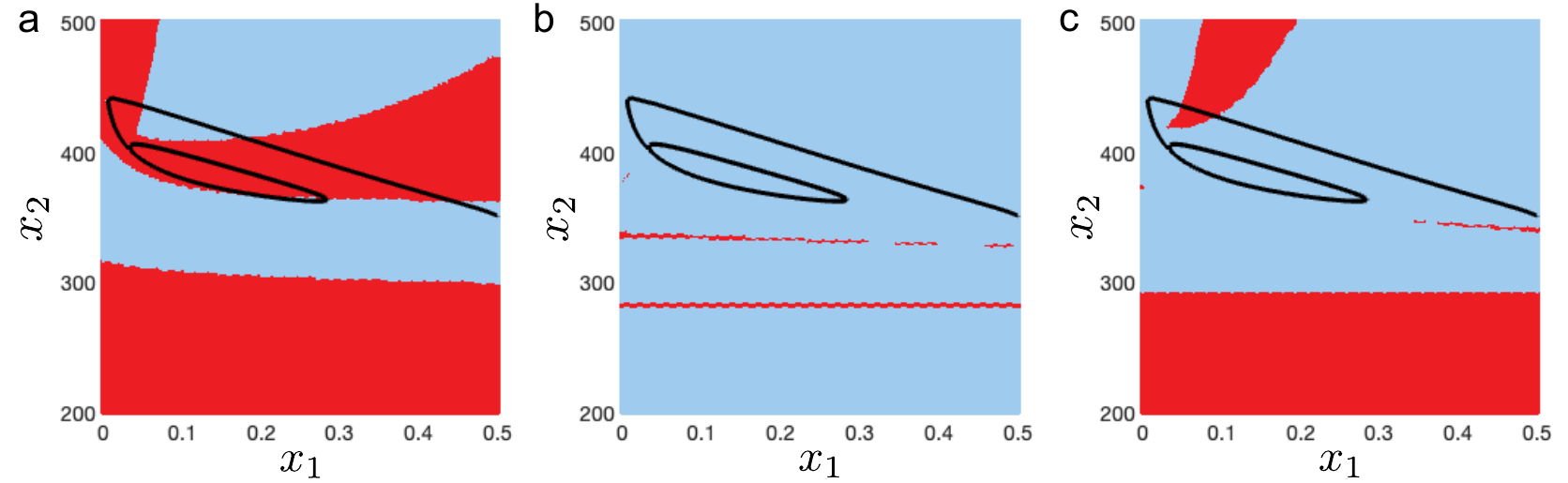}
    \caption{\label{fig.chemical}
    %
    Identifiable regions of the chemical reactor model for measurement signals given by: \textbf{(a)} $y = x_1$, \textbf{(b)} $y = x_2$, and \textbf{(c)} $y = x_3$. 
    The blue (red) colors correspond to states $\bm x$ in which the parameter $T_c$ is identifiable (unidentifiable). The black solid line represents the state trajectory $\bm x(t)$ starting at the initial condition $\bm x(t_0) = [0.5 \,\, 350 \,\, 0.4995]^\transp$ and converging to a limit cycle.
    The system parameters were set to $(c_0,T_0,T_c) = (1, 350,305)$, and $(k_1,k_2,k_3,k_4,k_5) = (1, 209.205,  2.0921, 8.7503\cdot 10^3,7.2\cdot 10^{10})$. 
    The tolerance for the numerical computation of the rank condition \eqref{eq.identifiabilitytest} is set to $\max_{\bm z} \sigma n \epsilon(\norm{\mathcal I(\theta,\bm z)}_2)$, where $\epsilon(b)$ is the floating-point relative accuracy of $b$. The numerical simulations are shown using the \texttt{ode15s} solver in MATLAB.
    }
\end{figure*}

\begin{rem} \label{rem.generic}
    Corollaries \ref{cor.idftlinear.semiexplict} and \ref{cor.idftlinear.ode} provide sufficient conditions for the \textit{local} identifiability of linear systems around any consistent state $\bm x\in\mathbb L$. The only exception is the origin $(\bm x,\dot{\bm x},\ldots,\bm x^{(\sigma)}) = 0$, where it follows trivially from Eq.~\eqref{eq.Athetax} that $[\bm x \,\, \dot{\bm x} \,\, \ddot{\bm x} \,\, \ldots \bm x^{(\mu)}]^\transp \otimes I_n = 0$ and hence condition \eqref{eq.rankAtheta} is not satisfied. This illustrates why parameter identification of a linear system is not possible when using data collected from a system operating on an equilibrium point; that is, the system dynamics were not ``excited'' during the data collection (as commonly referred in system identification).
\end{rem}


\section{\large Numerical Examples}
\label{sec.exmp}

This section demonstrates the identifiability analysis results on several descriptor systems, showing its applicability to experimental design and method validation in system identification.
Section \ref{example.chemical} presents a tutorial identifiability analysis for a simple index-1 chemical reactor model, showing step-by-step the computation of the proposed algebraic rank condition. This example highlights the parameter identifiability dependence on the choice of measured variables, also known as the sensor placement problem.
Moving up to an index-3 pendulum model, Section~\ref{example.pendulum} investigates the dependence of the parameter identifiability on the system state and experimental conditions (e.g., if the dataset is informative enough), shedding some light on the limitations of recent tools developed for parameter identification in DAEs \cite{Abdalmoaty2021}.
Finally, Section~\ref{example.linear} shows that even linear descriptor systems in which all state variables are independently measured may be unidentifiable, posing additional identification challenges compared to the special case of linear nonsingular systems.
Sections \ref{example.chemical} and \ref{example.pendulum} numerically evaluate the identifiability condition presented in Theorem \ref{theor.identifiability}, while Section \ref{example.linear} illustrates the results derived for linear systems proposed in Corollaries~\ref{cor.identifiability.linear}--\ref{cor.idftlinear.ode}.

\subsection{Chemical reactor}
\label{example.chemical}

Consider an exothermic reactor system with a single first-order reaction (\ce{A ->[$r(t)$] B}) and generated heat removed through an external cooling circuit. The chemical reactor is modeled by \cite{Ray}
\begin{subequations}%
\label{eq.chemicalreactor}%
    \begin{align}
    \dot{x}_1 &= k_1(c_0-x_1) - x_3 \label{eq.chemicalreactor.1} \\
    \dot{x}_2 &= k_1(T_0-x_2) + k_2x_3 + k_3(T_c-x_2), \\
    0 &= \rev{x_3 - k_5\exp{-{k_4}/{x_2}}x_1}, \label{eq.chemicalreactor.3}
    \end{align}
\end{subequations}
\noindent
where $\bm x = [x_1 \,\, x_2 \,\, x_3]^\transp\in\R^3$ is the state vector, $x_1$ is the concentration of reactant A, $x_2$ is the temperature, and $x_3$ is the reaction rate per unit volume (algebraic variable). Eq.~\eqref{eq.chemicalreactor.3} arises from a conservation law\footnote{This descriptor model can be directly converted to ODEs and then numerically integrated with conventional solvers (e.g., 4th-order Runge-Kutta). However, numerical solvers for DAEs enforce the constraint \eqref{eq.chemicalreactor.3} to be preserved in simulations like Fig.~\ref{fig.chemical}.} (Arrhenius equation) that imposes the dependence of reaction rates on temperature. The parameters are: inlet feed reactant concentration $c_0$ and feed temperature $T_0$, coolant temperature $T_c$, and other constants $k_i$.
Fig.~\ref{fig.chemical}a shows that the chemical reactor dynamics converge to a limit cycle for the considered set of parameters.

Let the coolant temperature be the parameter sought to be identified (i.e., $\theta=T_c$). The parameter estimation depends on the choice of observable (i.e., which state variable $x_i$ is measured by a sensor). To evaluate the appropriate sensor data for parameter estimation, we now test which observables make the system identifiable.
As a tutorial, we analyze step-by-step the identifiability of the descriptor model \eqref{eq.chemicalreactor} considering the output signal $y(t) = x_1(t)$, $t\in[0,10]$ s. First, we augment model \eqref{eq.chemicalreactor} by representing $\theta(t) = T_c(t)$ as an additional state variable with constant dynamics, i.e., $\dot\theta(t) = 0$. Functions $\bar{\bm H}_{\nu}$ and $\bar{\bm F}_{\mu}$ are then constructed up to order $\mu=\nu=n-1=2$:
$\bar{\bm H}_{\nu} = [x_1 \,\, \dot x_1 \,\, \ddot x_1 ]^\transp$ and $\bar{\bm F}_{\mu} = [
        (\bm F)^\transp \,\, (\frac{{\rm d}\bm F}{{\rm d}t})^\transp \,\, (\frac{{\rm d}^2\bm F}{{\rm d}t^2})^\transp]^\transp$, where
\begin{align*}
    \bm F &=
    \begin{bmatrix}
        k_1(c_0-x_1) - x_3 - \dot x_1 \\
        k_1(T_0-x_2) + k_2x_3 + k_3(\theta - x_2) - \dot x_2\\
        x_3 - k_5\exp{-\frac{k_4}{x_2}}x_1
    \end{bmatrix},
\\
    \frac{{\rm d}\bm F}{{\rm d}t} &=
    \begin{bmatrix}
        \rev{-k_1\dot x_1 - \dot x_3 - \ddot x_1} \\
        -k_1\dot x_2 + k_2 \dot x_3 + k_3(\dot\theta - \dot x_2) - \ddot x_2 \\
        \rev{-\dot x_3 - k_5\exp{-\frac{k_4}{x_2}}\left(\dot x_1 + k_4\frac{x_1\dot x_2}{x_2^2}\right)}
        \end{bmatrix},
\end{align*}
\noindent
and ${\rm d}^2\bm F/{\rm d}t^2$ is omitted for the sake of brevity. Note that $\bar{\bm F}_{\mu}$ and $\bar{\bm H}_{\nu}$ are functions of $\theta$, $\bm x$, and its time derivatives. By assumption, $\dot\theta = \ddot\theta = \dddot\theta = 0$, $\forall t$. After building the identifiability matrix, condition \eqref{eq.identifiabilitytest} can then be evaluated to test if the system is identifiable at any consistent state $\bm x_0\in\mathbb L\subset\R^3$ and parameter value $\theta_0\in\R^1$.

Fig.~\ref{fig.chemical} shows the identifiable regions in the state space depending on the observable choice. For the output $y(t)=x_1(t)$ (Fig.~\ref{fig.chemical}a), the limit cycle solution ${\bm \Phi}_T(\bm x(t_0);\theta)$ predominantly lies in an unidentifiable region. At unidentifiable states $\bm x_0$, measured trajectories $\bm h\circ{\bm\Phi}_T(\bm x_0;\theta)$ sufficiently close to $\bm x_0$ are \textit{not} guaranteed to be distinguishable from a trajectory $\bm h\circ{\bm\Phi}_T(\bm x_0;\theta')$ corresponding to some distinct parameter $\theta'\neq\theta$ lying in a neighborhood $\mathcal U$ of the true parameter value $\theta$. In other words, there may exist some $(\theta'\in\mathcal U)\neq\theta$ such that $\bm h\circ{\bm\Phi}_T(\bm x_0;\theta) = \bm h\circ{\bm\Phi}_T(\bm x_0;\theta')$. 
The predominance of unidentifiable regions in the state space implies the possible existence of non-unique solutions for $\theta$ in the neighborhood of ${\bm \Phi}_T(\bm x(t_0);\bm\theta)$. This directly impairs the estimation accuracy of $\theta$ if the dataset mostly include measurement data collected in these unidentifiable states. 

Despite the predominance of unidentifiable regions in Fig. \ref{fig.chemical}a, parameter identifiability is achieved periodically in a small region of the limit cycle around $\bm x(t) = [0.25 \,\, 360 \,\, 0.5]^\transp$ and other regions containing part of the transient response. 
Theoretically, parameter estimation could be possible by using data collected only at the identifiable states. However, this may not be possible in practice if the identifiable region is too small as in Fig. \ref{fig.chemical}a, either due to limited data availability in these regions or ill-conditioned identifiability matrices at states very close to unidentifiable regions. If the numerical rank tolerances are not appropriately chosen, the rank condition \eqref{eq.identifiabilitytest} may be numerically satisfied for ill-conditioned matrices $\mathcal I(\bm \theta,\bm z)$ even though the system is practically unidentifiable.

After conducting the same identifiability analysis for $y(t) = x_2(t)$ and $y(t) = x_3(t)$, Fig.~\ref{fig.chemical}b,c shows that the system is locally identifiable everywhere around the limit cycle solution when these other types of measurement signals (sensors) are considered. This is a particularly beneficial result for this sensor placement problem since temperature sensors ($y = x_2$) are very affordable and practical to install in chemical reactors.


\begin{figure*}[t]
    \centering
    \includegraphics[width=0.95\linewidth]{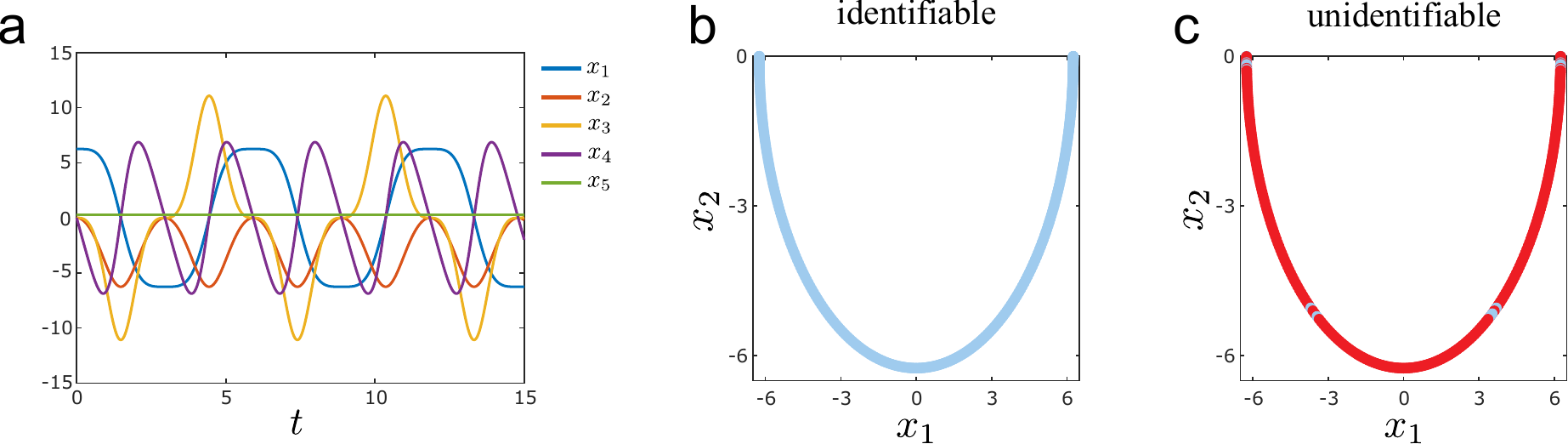}
    \caption{\textbf{(a)} State evolution of the pendulum equation as a function of time. \textbf{(b, c)} Identifiable (blue) and unidentifiable (red) regions in the phase plane $(x_1,x_2)$ depending on the parameter set $\theta$ sought to be identified: (b) parameter sets $m$, $g$, $L$, $[m \,\, g]$, $[g \,\, L]$, and $[m \,\, L]$ are all (locally) identifiable everywhere in the phase space; (c) parameter set $[m \,\, g \,\, L]$ is unidentifiable almost everywhere in the phase space. The parameters are set in the simulation as $L=6.25$, $g=9.81$, and $m=0.3$.}
    \label{fig.pendulum}
\end{figure*}

\begin{figure*}[t!]
    \centering
    \includegraphics[width=0.92\linewidth]{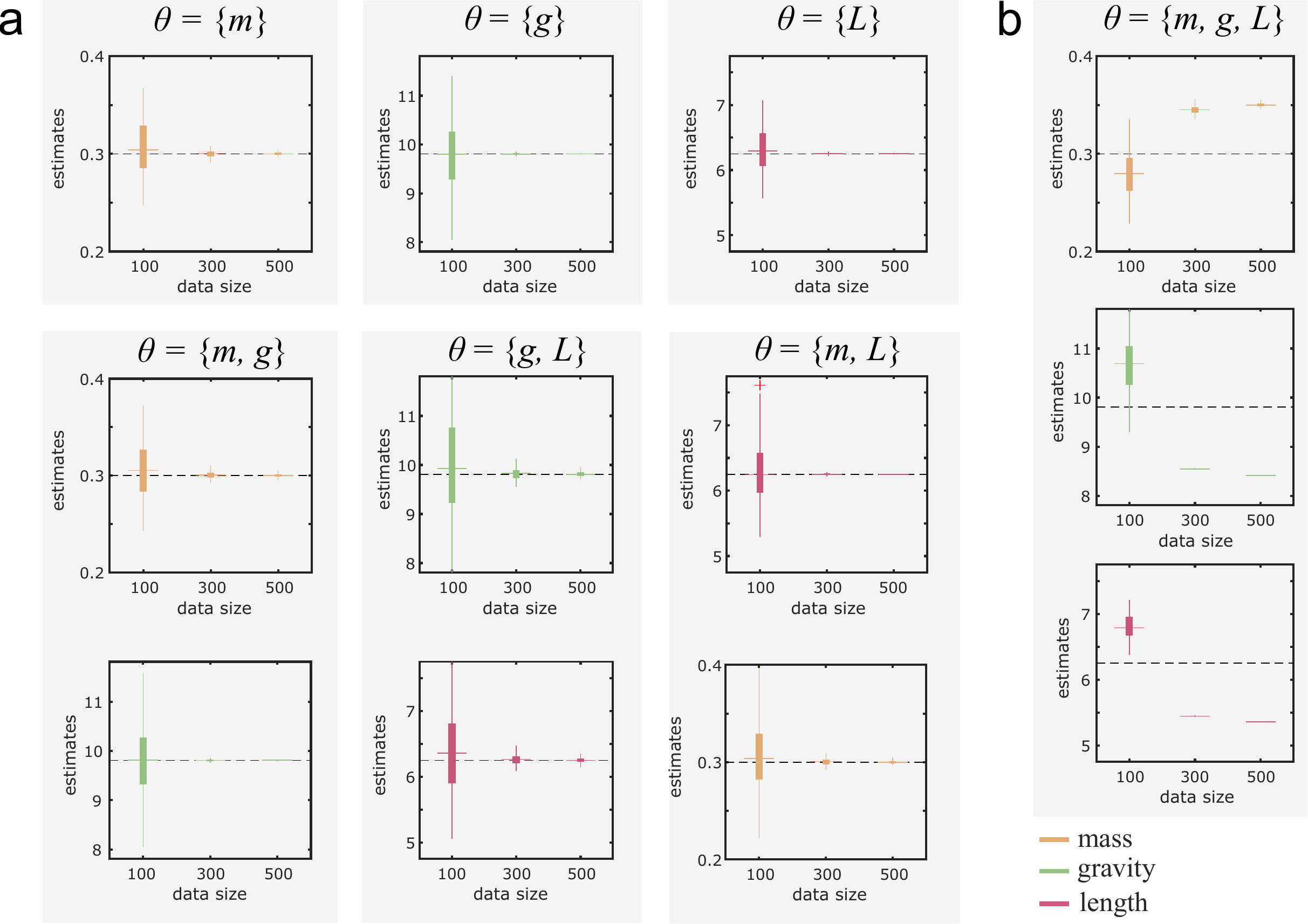}
    \caption{\label{fig.pendulum.estimation} Parameter estimation of the pendulum model using a prediction-error method for nonlinear DAE identification. Boxplots of the estimated parameter values for the sets of parameters predicted to be: \textbf{(a)} identifiable versus \textbf{(b)} unidentifiable. The x-axis shows the number of data points used for the parameter estimation. Each box comprise the estimated values for a 100 independent experiments. The dashed lines indicate the true values of the parameters. In each experiment, the mean square error between a measured noisy trajectory and a trajectory predicted by the model is minimized over the model parameters. The minimization is performed using the Levenberg-Marquardt algorithm, using the same settings as in the paper proposing the method \cite{Bereza2022}. The initial estimates for the parameters are set to $m'=0.1$, $g'=8.21$, and $L'=5.41$.
    }
\end{figure*}

\subsection{Pendulum equation}
\label{example.pendulum}

The motion of an undamped pendulum is represented in Cartesian coordinates as an index-3 DAE system:
\begin{subequations}
\label{eq.pendulum}
    \begin{align}
    \dot{x}_1 &= x_3, \label{eq.pendulum.1} \\
    \dot{x}_2 &= x_4, \\
    \dot{x}_3 &= \frac{x_1 x_5}{m}, \\
    \dot{x}_4 &= \frac{x_2 x_5}{m} - g, \label{eq.pendulum.4} \\
    0 &= x_1^2 + x_2^2 - L^2, \label{eq.pendulum.5}
    \end{align}
\end{subequations}
\noindent
where $(x_1,x_2)$ denote the Cartesian coordinates of the pendulum endpoint in the plane, $(x_3,x_4)$ are the corresponding velocities, and $x_5$ is the tension per unit length in the pendulum arm. The parameters are: mass $m$, gravitational force $g$, and pendulum-arm length $L$.
Assume that only the pendulum angle is measured, i.e.,
\begin{equation}
    y = h(\bm x) = \arctan\left(-\frac{x_1}{x_2}\right).
\label{eq.pendulum.output}
\end{equation}
\noindent
We now test the identifiability of the pendulum's individual parameters ($\theta = m$, $g$, or $L$) as well as all its possible combinations (e.g., $\bm\theta = [g \,\, L]$). The parameter identifiability is evaluated for all (consistent) states in the simulated trajectory presented in Fig.~\ref{fig.pendulum}a.
Figure~\ref{fig.pendulum}b shows that the system is locally identifiable with respect to all individual and pairwise combination of parameters for the entire set of consistent states $(x_1,x_2)\in\mathbb L$. On the other hand, the set of all parameters $\theta = [m \,\, g \,\, L]$ is unidentifiable\footnote{The identifiability condition is formally only sufficient. However, given that $\mathbb L$ is dense, it follows from Remark~\ref{rem.almostnecessary} that the condition is also necessary except for a local set of parameters with Lesbegue dimension zero.} for almost every $\bm x\in\mathbb L$ (Fig.~\ref{fig.pendulum}c), highlighting that\textemdash in practice\textemdash we expect that at least one of the parameters must be known a priori before methods for parameter estimation can be successfully applied. 

\begin{figure*}[t] 
    \centering
    \includegraphics[width=0.8\linewidth]{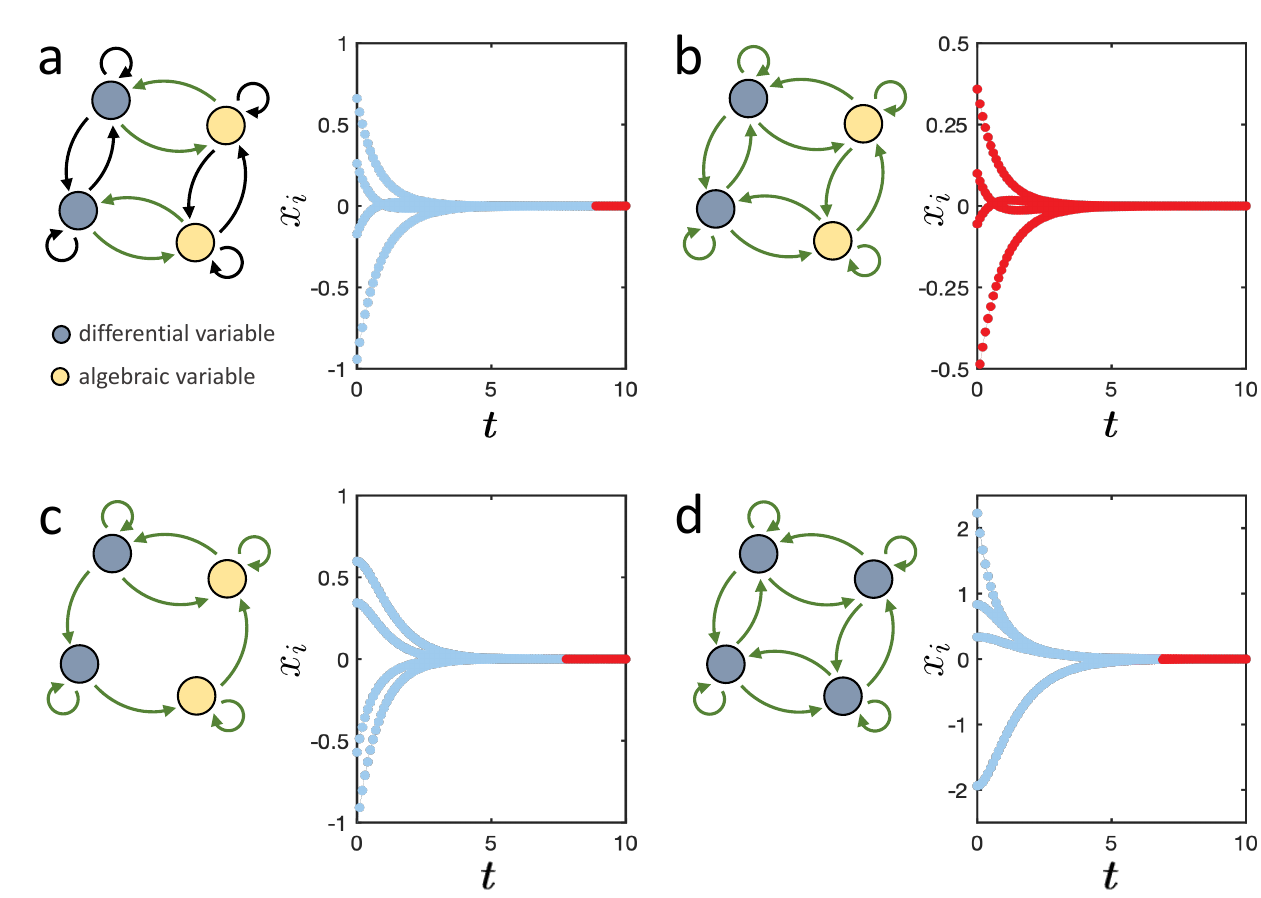}
    \caption{\label{fig.lineardae} Identifiability of linear descriptor systems for different model structures and parameter sets. 
    The network structure of the underlying model is shown on the left panels: each node represents a (differential or algebraic) variable and an edge points from node $j$ to $i$ if $A_{ij}\neq 0$. The parameters (represented by edges) sought to be identified in each case are highlighted in green.
    The identifiable (unidentifiable) states in the simulated trajectory are indicated in blue (red) on the right panels.
    The initial conditions of differential variables $\bm x_1(t_0)$ were randomly drawn from a normal distribution.
    \textbf{(a)} Identifiability of submatrices $A_{12}$ and $A_{21}$ in a dense DAE model.
    \textbf{(b)}  Identifiability of matrix $A$ in a dense DAE model.
    \textbf{(c)} Identifiability of matrix $A$ in a sparse DAE model. 
    \textbf{(d)} Identifiability of matrix $A$ in a dense ODE model.
    }

\end{figure*}

The above theoretical analysis of the parameters identifiability is supported by computational results using a prediction-error method developed for parameter estimation of general nonlinear descriptor systems \cite{Abdalmoaty2021,Bereza2022}. 
Figure \ref{fig.pendulum.estimation} shows that \textit{identifiable} parameter sets can indeed be accurately estimated using the prediction-error method; the parameter estimates converge to the true values as the number of data samples increases. On the other hand, the estimation of the parameter set $\theta=[m\,\, g \,\, L]$, which is predicted to be unidentifiable by our theory, diverged from the true values. 
Although such biased estimation in Fig. \ref{fig.pendulum.estimation}b could be mistakenly attributed to numerical issues, poor data, or lack of scalability of prediciton-error methods to larger sets of parameters, our analysis highlights that the failed parameter estimation is actually a direct consequence of the identifiability of the model structure.



\subsection{Linear DAE system}
\label{example.linear}

Consider the linear descriptor system \eqref{eq.lineardescriptor.augmented} of order $n=4$. Let $E = \operatorname{diag}(I_2,0_{2\times 2})$, leading to $n_1 = 2$ differential and $n_2 = 2$ algebraic variables. Assume that all state variables are independently measured, i.e., $C = I_4$. We choose a random system matrix such that the system is asymptotically stable:
\begin{equation*}
    A
    =
    \left[\begin{array}{cc|cc}
        -1.72 &  0.20 &  0.02 & -0.23 \\
        0.15  & -1.30 & -0.31 & -0.21 \\
        \hline
        0.51  & -0.87 & -1.50 & -0.21 \\
        -1.05 & -0.20 & -0.23 & -1.34
        \end{array}\right].
\end{equation*}

\smallskip
Figure \ref{fig.lineardae} shows the identifiability analysis of linear systems for different choices of parameters and model structures. Note that $A_{22}$ is invertible and hence the descriptor system is index 1 (Assumption \ref{assump.index1lineardae}). Since $\rank(A_{21})=2$, it follows from Corollary \ref{cor.idftlinear.semiexplict} that the set of parameters $\bm\theta_1 = \operatorname{vec}(A_{12},A_{21})$ is identifiable. As predicted by the theory, Fig. \ref{fig.lineardae}a shows that the system is indeed locally identifiable with respect to $\bm\theta_1$ at the states $\bm x(t)$ in the simulated trajectory. In these simulations, the identifiable states are determined by evaluating the rank condition \eqref{eq.identcond.linear} at state $\bm x(t)$. Only when the system converges to equilibrium ($\bm x(t)\rightarrow 0$) that identifiability is lost\footnote{Theoretically, identifiability is lost exactly at $\bm x(t) = 0$. In practice, the rank condition \eqref{eq.identcond.linear} holds only up to some numerical rank tolerance as $\bm x(t)$ approaches zero (Fig. \ref{fig.lineardae}a,c,d).}, as pointed in Remark \ref{rem.generic}. Similar results can be verified for the identifiability of $\bm\theta = \operatorname{vec}(A_{11},A_{22})$.

For the parameter set $\bm\theta_2 = \operatorname{vec}(A)$, Fig. \ref{fig.lineardae}b shows that the system is \textit{not} identifiable. This demonstrates that, for a generic descriptor system, full-state measurement is not enough for identifiability of the system matrix $A$. Prior knowledge of the system is thus essential to ascertain identifiability. One alternative is to reduce the set of parameters to be identified (as in Fig. \ref{fig.lineardae}a), but this indirectly assumes that all other parameters are already known a priori. Another approach is to explore the model sparsity and prior knowledge of zero entries in matrix $A$; the interconnection structure is often reliably known in many network systems. Consider, for example, the sparse model in which the following matrix entries are zero: $\{a_{12}, a_{13},a_{23},a_{24},a_{31},a_{34},a_{42}\}$. Figure \ref{fig.lineardae}c illustrates that identifiability of the sparse system with respect to $\bm\theta_2$ can indeed be achieved. Beyond the results presented in Corollary \ref{cor.idftlinear.semiexplict}, we expect that generic identifiability conditions based on model sparsity can be derived, although this is a topic of future research.

If $E=I_n$ (leading to an ODE system of order $4$), it follows from Corollary \ref{cor.idftlinear.ode} that the system is identifiable with respect to $\bm\theta_2$. This statement is verified numerically in Fig. \ref{fig.lineardae}d. The contrast between Figs. \ref{fig.lineardae}b and d demonstrates, even in the linear setting, DAEs are fundamentally more challenging to identify from data than ODEs.


\section{\large Conclusion}
\label{sec.conc}

The large adoption of Kalman's controllability and observability tests can be attributed to its simple formulation and easy implementation involving only matrix multiplications and prior knowledge of the system model. These tests can be intuitively examined based on the system structure \cite{Lin1974a,Boukhobza2006observability,Montanari2020} and have since been extended to more general classes of systems \cite{Hermann1977,Chen1980,Campbell1991duality,Montanari2022nonlinear}. Our results are motivated by a similar goal: to derive an identifiability condition that can be straightforwardly tested without excessive manipulation of system equations. The proposed identifiability test (Theorem~\ref{theor.identifiability}) requires only prior knowledge of the system's implicit function $\bm F(\bm x,\dot{\bm x})$, measurement function $\bm h(\bm x)$, and their successive derivatives (which can be obtained via symbolic computation). Codes for the identifiability test of general descriptor systems $\{\bm F,\bm h\}$ are available in GitHub (\url{https://github.com/montanariarthur/IdentifiabilityDAE}), along with implementations of the examples discussed in Section~\ref{sec.exmp}. We expect that the presented results will find a wide range of DAE applications beyond the examples considered in this work, such as singular systems, power systems, multi-body robotics, biological regulatory networks, and many others. In particular, by treating feedback control laws $\bm u = \bm k(\bm x)$ as algebraic equations \eqref{eq.algebraicsys}, the proposed identifiability test may also provide insight in the limitations of closed-loop identification \cite{van1998closed,forssell1999closed}.

There are still many challenges for the data-driven modeling of descriptor systems. First, it still remains unclear which types of model structures and measurement functions are sufficient to guarantee the system identifiability, even in the linear case as explored in Section \ref{example.linear}. We expect that methods based on structural network inference \cite{Gonccalves2008necessary}, graph-theoretical conditions for observability \cite{Lin1974a,Montanari2020,Montanari2023target}, or rank properties of the Kronecker product \cite{tian2005some} might be able to address these problems. Second, the identifiability analysis is grounded on the assumption that the parameters can be modeled as \textit{time-constant} state variables. However, whether the identifiability of a DAE system is promoted or inhibited by time-dependent parameters (i.e., $\dot{\bm\theta} = \bm f(t)$) has yet to be explored. Finally, our results can be directly extended to discrete-time systems and feedback systems (Remark~\ref{rem.feedbacksys}), but generalizations to broader classes of systems such as stochastic DAEs \cite{WINKLER2004435} or partial DAEs \cite{Ben14} remain an open problem.

\section*{\normalsize Acknowledgement}                          
The authors thank Mohamed R.-H. Abdalmoaty and Håkan Hjalmarsson for insightful discussions. FL was supported by the Luxembourg National Research Fund (FNR), grant CORE19/13684479/DynCell, and is now a Quilvest Research Fellow under a Quilvest donation. RB is supported by the Swedish Research Council under contracts 2019-04956 and 2016-06079 (the research environment NewLEADS) and by Digital Futures.

\appendix
\section*{\large Appendix: Proofs}

\label{app.proofs}

\textbf{PROOF OF PROPOSITION~\ref{corol.obsvnonsingular}. }
\textit{Sufficiency of condition 1.} Let $\mathcal O(\bm x,\bm w)$ be the corresponding observability matrix for a nonsingular system $\bm F(\bm x,\dot{\bm x}) = \bm f(\bm x) - \dot{\bm x}$ and measurement function $\bm h(\bm x)$. Note that $\partial\bar{\bm F}_\mu/\partial\bm w$ is a block lower triangular matrix with diagonal matrices $-I_n$. Thus, it follows from condition \eqref{eq.daeobsv} in Theorem~\ref{theor.nonlinearobsv} that the pair $\{\bm F,\bm h\}$ is observable if $\rank(\mathcal O(\bm x,\bm w)) = n + n(\mu + 1)$ or, equivalently, $\mathcal O(\bm x,\bm w)$ is invertible.

\textit{Equivalence between conditions 1 and 2.} Since $\partial\bar{\bm F}_\mu/\partial\bm w$ is invertible, then $\mathcal O(\bm x,\bm w)$ is invertible if and only if
\begin{equation*}
    \mathcal O'(\bm x,\bm w) = 
    \pdv{\bar{\bm H}_\nu}{\bm x} 
    - 
    \pdv{\bar{\bm H}_\nu}{\bm w}
    \left(
    \pdv{\bar{\bm F}_\mu}{\bm w}
    \right)^{-1}
    \pdv{\bar{\bm F}_\mu}{\bm x}
\end{equation*}
\noindent
is invertible. We now show that $\mathcal O'(\bm x,\bm w) = \bm\Psi(\bm x)$, therefore implying that conditions 1 and 2 are equivalent. We show this, without loss of generality, for the univariate case $f(x)$ and $h(x)$, where $n=1$. After some matrix manipulations, it follows that
\begin{align*}
    \mathcal O' &=
    \begin{bmatrix}
        \pdv{h}{x}
        \\
        \pdv[2]{h}{x}\dot x + \pdv{h}{x}\pdv{f}{x}
        \\
        \pdv[3]{h}{x}\dot x^2 + \pdv[2]{h}{x}\ddot x + 2\pdv[2]{h}{x}\pdv{f}{x}\dot x + \pdv{h}{x}\left(\pdv{f}{x}\right)^2 + \pdv{h}{x}\pdv[2]{f}{x}\dot x
        \\
        \vdots
    \end{bmatrix}
    \\
    & =
    \begin{bmatrix}
        \pdv{h}{x}
        \\
        \pdv{}{x}\left(\pdv{h}{x} f\right)
        \\
        \pdv{}{x}\left(\pdv[2]{h}{x} f^2 + \pdv{h}{x}\pdv{f}{x} f\right)
        \\
        \vdots
    \end{bmatrix}
    =
    \rev{\pdv{}{\bm x}}
    \begin{bmatrix}
        \mathcal L_f^0 h(x)
        \\
        \mathcal L_f^1 h(x)
        \\
        \mathcal L_f^2 h(x)
        \\
        \vdots
    \end{bmatrix}
    = \bm\Psi(x),
\end{align*}
\noindent
where we have substituted $\dot x = f(x)$ and $\ddot x = \pdv{f(x)}{x} f(x)$ and the equality holds by induction for all higher-order terms $\mathcal L_f^\nu h(x)$, $\nu\geq 3$.   \QEDS

\noindent
\textbf{PROOF OF PROPOSITION~\ref{corol.lineardescriptor}. } 
%
%
\textit{Sufficiency of statement 2.} The sufficiency of statement 2 follows directly from Theorem~\ref{theor.nonlinearobsv} by computing the observability matrix $\mathcal O(\bm x,\bm w)$ corresponding to the linear functions $\bm F(\bm x,\dot{\bm x}) = A\bm x - E\dot{\bm x}$ and $\bm h(\bm x) = C\bm x$ up to order $\sigma = n$. In the linear case, condition \eqref{eq.daeobsv} becomes globally valid since $\mathcal O(\bm x,\bm w) = O$, which has constant rank for every state $\bm x\in\mathbb L$. It follows from Lemma~\ref{lemma.1full} that if statement 2 holds, then $O$ is 1-full with respect to $\bm x$. Therefore, any arbitrary $\bm x(t_0)\in\mathbb L$ can be uniquely reconstructed from $\bm y(t)$ (and its successive derivatives up to order $n-1$) and, by Definition~\ref{def.linearobsv}, the system is R-observable. 

\textit{Necessity of statement 2.} Since statement 3 is sufficient and necessary condition for R-observability, we now prove the necessity of statement 2 by showing that statement 3 entails statement 2. It follows from statement 2 that 
\begin{align*}
    \rank(\mathcal{O}) &= 
    \rank
    \begin{bmatrix}
     A\\
     C
    \end{bmatrix}
     + \rank
    \begin{bmatrix}
        \mathcal{O}_{12}\\
        \mathcal{O}_{22}
    \end{bmatrix}\\
    &= n + \rank
    \begin{bmatrix}
        \mathcal{O}_{12}\\
        \mathcal{O}_{22}
    \end{bmatrix}.
\end{align*}
\noindent
Observe that $\rank[A^\transp \,\, C^\transp]^\transp = n$ is equivalent to statement 3 for $\lambda = 0$, which concludes the proof.  \QEDS
\noindent
\textbf{PROOF OF THEOREM~\ref{theor.identifiability}. } 
Consider the extended system \eqref{eq.implicitdescriptor.augmented} and recall that $\mu$ and $\nu$ respectively define the higher-order derivatives of functions $\bm F$ and $\bm H$ in Eq. \eqref{eq.subsequentdifferentiation}. Subsequent differentiation of Eq. \eqref{eq.implicitdescriptor.augmented} yields
\begin{equation} \label{eq.proof.FThetaH}
    \begin{bmatrix}
        \bar{\bm F}_{\mu}(\bm x,\ldots,\bm x^{(\mu + 1)},\bm\theta)
        \\
        \bm\Theta
        \\
        \bar{\bm H}_{\nu}(\bm x,\ldots,\bm x^{(\nu)},\bm\theta)
    \end{bmatrix}
    =
    \begin{bmatrix}
        0 \\ 0 \\ \bar{\bm y}
    \end{bmatrix},
\end{equation}
\noindent
where $\bm\Theta = [\dot{\bm\theta}^\transp \,\, \ddot{\bm\theta}^\transp \,\, \ldots \,\, \bm\theta^{(\mu+1) ^\transp}]^\transp$. We define the function $\bm\Psi(\bm\theta,\ldots,\bm\theta^{(\sigma)},\bm x,\ldots,\bm x^{(\sigma)}) := [\bar{\bm F}_{\mu}^\transp \,\, \bm\Theta^\transp \,\, \bar{\bm H}_{\nu}^\transp]^\transp$. The Jacobian matrix of ${\rm D}\bm\Psi$ is given by
\begin{equation*}
\begin{aligned}
    {\rm D}\bm\Psi &= 
    \begin{bmatrix}
        \dfrac{\partial \bar{\bm F}_{\mu}}{\partial \bm \theta}
        &
        \dfrac{\partial \bar{\bm F}_{\mu}}{\partial (\bm x,\dot{\bm x},\ldots,\bm x^{(\sigma)})}
        &
        \dfrac{\partial \bar{\bm F}_{\mu}}{\partial \bm (\dot{\bm\theta},\ddot{\bm\theta},\ldots,\bm\theta^{(\sigma)})}
        \\
        \dfrac{\partial\bm\Theta}{\partial\bm\theta}
        &
        \dfrac{\partial\bm\Theta}{\partial (\bm x,\dot{\bm x},\ldots,\bm x^{(\sigma)})}
        &
        \dfrac{\partial\bm\Theta}{\partial \bm (\dot{\bm\theta},\ddot{\bm\theta},\ldots,\bm\theta^{(\sigma)})}
        \\
        \dfrac{\partial \bar{\bm H}_{\nu}}{\partial \bm \theta}
        &
        \dfrac{\partial \bar{\bm H}_{\nu}}{\partial (\bm x,\dot{\bm x},\ldots,\bm x^{(\sigma)})}
        &
        \dfrac{\partial \bar{\bm H}_{\nu}}{\partial \bm (\dot{\bm\theta},\ddot{\bm\theta},\ldots,\bm\theta^{(\sigma)})}
    \end{bmatrix}
    \\
    &=
    \begin{bmatrix}
        \dfrac{\partial \bar{\bm F}_{\mu}}{\partial \bm \theta}
        &
        \dfrac{\partial \bar{\bm F}_{\mu}}{\partial \bm z}
        &
        0_{n(\mu+1)\times p\sigma}
        \\
        0_{p(\mu+1)\times p}
        &
        0_{p(\mu+1)\times n(\sigma+1)}
        &
        \Delta
        \\
        \dfrac{\partial \bar{\bm H}_{\nu}}{\partial \bm \theta}
        &
        \dfrac{\partial \bar{\bm H}_{\nu}}{\partial \bm z}
        &
        0_{q(\nu+1)\times p\sigma}
    \end{bmatrix}
\end{aligned}
\end{equation*}
\noindent
where $\Delta = [I_{p(\mu+1)\times p(\mu+1)} \,\, 0_{p(\mu+1)\times p(\sigma-\mu-1)}]$.

Now, consider the following linear map of Eq.~\eqref{eq.proof.FThetaH} evaluated locally around $(\bm\theta,\bm z,\bm w)$:
\begin{equation} \label{eq.proof.linearmap}
    {\rm D}\bm\Psi 
        \begin{bmatrix}
            \bm \theta \\
            \bm z \\
            \bm w
        \end{bmatrix}
        =
        \begin{bmatrix}
        0_{n(\mu+1)\times 1} \\ 0_{p(\mu+1)\times 1} \\ \bar{\bm y}
        \end{bmatrix},
\end{equation}
\noindent
where $\bm w = [\dot{\bm\theta}^\transp \,\, \ddot{\bm\theta}^\transp \,\, \ldots \,\, \bm\theta^{(\sigma)^\transp}]^\transp$.
By Lemma \ref{lemma.1full}, the linear system \eqref{eq.proof.linearmap} is 1-full with respect to $\bm \theta$ if and only if $\rank(\partial\bm\Psi/\partial\bm\theta) = p$ and $\operatorname{col}(\partial\bm\Psi/\partial\bm\theta)\cap\operatorname{col}(\partial\bm\Psi/\partial(\bm z,\bm w))=\emptyset$. Clearly, $\partial\bm\Psi/\partial\bm\theta$ is full column rank if and only if
\begin{equation} \label{eq.proof.M1isfullrank}
    \rank 
    \begin{bmatrix}
        \dfrac{\partial \bar{\bm F}_{\mu}}{\partial \bm \theta}
        \\
        \dfrac{\partial \bar{\bm H}_{\mu}}{\partial \bm \theta}
    \end{bmatrix} = p.
\end{equation}
\noindent
Moreover, note that $\operatorname{col}(\partial\bm\Psi/\partial\bm\theta)\cap\operatorname{col}(\partial\bm\Psi/\partial\bm w)$ is empty for any system $\{\bm F,\bm h\}$. Therefore, it follows that the 1-fullness of system \eqref{eq.proof.linearmap} with respect to $\bm \theta$ is equivalent to Eq. \eqref{eq.proof.M1isfullrank} and 
$\operatorname{col}(\partial\bm\Psi/\partial\bm\theta)\cap\operatorname{col}(\partial\bm\Psi/\partial\bm z)=\emptyset$, which in turn are equivalent to condition \eqref{eq.identifiabilitytest} according to Lemma~\ref{lemma.1full}.

Recall that $\bm\Psi : \mathcal M \mapsto \mathcal N$ is assumed to be smooth, where $\mathcal M\subset \R^{(\sigma+1)(n+p)}$ and $\mathcal N\subset \R^{(\mu+1)(n+p)+n\nu}$. Let ${\rm T}_{\bm p} \mathcal M$ denote the tangent space of $\mathcal M$ around a point $\bm p = [\bm\theta_0^\transp \,\, \bm z_0^\transp \,\, \bm w_0^\transp]^\transp \in\mathcal M$. Following the Constant Rank Theorem, there exist open neighborhoods $\mathcal U$ around $\bm p$ and $\mathcal V$ around $\bm\Psi(\bm p)$, as well as some diffeomorphisms $\bm\phi : {\rm T}_{\bm p}\mathcal M\mapsto \mathcal U$ and $\bm\psi : {\rm T}_{\bm\Psi(\bm p)} \mathcal N \mapsto \mathcal V$, such that $\bm\Psi(\mathcal U)\subseteq \mathcal V$ and the Jacobian ${\rm D}\bm\Psi : {\rm T}_{\bm p}\mathcal M \mapsto {\rm T}_{\bm\Psi(\bm p)}$ is given by ${\rm D}\bm\Psi = \bm\psi^{-1} \circ \bm \Psi \circ \bm \phi$.

The proof now follows from the fact that, if condition~\eqref{eq.identifiabilitytest} holds, then the system \eqref{eq.proof.linearmap} is 1-full with respect to $\bm\theta$. This implies that $\bm\theta$ can be uniquely determined for any consistent measurement $\bar{\bm y}$ and known matrix $D\bm\Psi$. Finally, the equivalence ${\rm D}\bm\Psi = \bm\psi^{-1} \circ \bm \Psi \circ \bm \phi$ shows that $\bm\theta$ can also be uniquely determined from $\bar{\bm y}$ and the system model $\bm\Psi$. Therefore, the system is locally identifiable around a neighborhood $\mathcal U$ of $(\bm\theta_0,\bm x_0)$ and its successive derivatives up to order $\sigma$.  \QEDS

\noindent
\textbf{PROOF OF COROLLARY~\ref{cor.idftlinear.semiexplict}. } 
Consider the general identifiability analysis where $\bm\theta = \operatorname{vec}(A_{11},A_{12},A_{21},A_{22})$. Denote 
    \begin{equation*}
    \bm\Gamma = 
        \begin{bmatrix}
                A(\bm\theta)\bm x \\
                A(\bm\theta)\dot{\bm x} \\
                \vdots \\
                A(\bm\theta)\bm x^{(\mu)}
        \end{bmatrix}
        \,\,\, \text{and} \,\,\,
    \bm\Gamma_{ij} = 
        \begin{bmatrix}
                A_{ij}(\bm\theta)\bm x_j \\
                A_{ij}(\bm\theta)\dot{\bm x}_j \\
                \vdots \\
                A_{ij}(\bm\theta)\bm x^{(\mu)}_j
        \end{bmatrix}.
    \end{equation*}
    \noindent
    As in Eq.~\eqref{eq.Athetax}, it follows that
    \begin{equation*}
    \begin{aligned}
        \mathcal I_{11} &= 
        \begin{small}
        \begin{bmatrix}
            \dfrac{\partial \bm\Gamma}{\partial \operatorname{vec}(A_{11})}
            &
            \dfrac{\partial \bm\Gamma}{\partial \operatorname{vec}(A_{12})}
            &
            \dfrac{\partial \bm\Gamma}{\partial \operatorname{vec}(A_{21})}
            &
            \dfrac{\partial \bm\Gamma}{\partial \operatorname{vec}(A_{22})}
        \end{bmatrix}
        \end{small}
        \\
        &= 
        \begin{small}
        \begin{bmatrix}
            \dfrac{\partial \bm\Gamma_{11}}{\partial \operatorname{vec}(A_{11})}
            &
            \dfrac{\partial \bm\Gamma_{12}}{\partial \operatorname{vec}(A_{12})}
            &
            0
            &
            0
            \\
            0
            &
            0
            &
            \dfrac{\partial \bm\Gamma_{21}}{\partial \operatorname{vec}(A_{21})}
            &
            \dfrac{\partial \bm\Gamma_{22}}{\partial \operatorname{vec}(A_{22})}
        \end{bmatrix}
        \end{small},
    \end{aligned}
    \end{equation*}
    \noindent
    where $\dfrac{\partial \bm\Gamma_{ij}}{\partial \operatorname{vec}(A_{ij})} = [\bm x_j \,\, \dot{\bm x}_j \,\, \ddot{\bm x}_j \,\, \ldots \bm x_j^{(\mu)}]^\transp \otimes I_{n_i}$.
    %

    For the case of $\bm\theta = \operatorname{vec}(A_{11},A_{22})$, it holds that
    \begin{subequations}
    \begin{align*}
        \rank(\mathcal I_{11}) &= \rank\left(\dfrac{\partial \bm\Gamma_{11}}{\partial \operatorname{vec}(A_{11})} \right) + \rank\left(\dfrac{\partial \bm\Gamma_{22}}{\partial \operatorname{vec}(A_{22})} \right)
        \\
        &=
        \rank\left([\bm x_1 \,\, \dot{\bm x}_1 \,\, \ddot{\bm x}_1 \,\, \ldots \bm x_1^{(\mu)}]^\transp\right)\rank(I_{n_1})
        \\
        &\,\,\,\,\,\, + \rank\left([\bm x_2 \,\, \dot{\bm x}_2 \,\, \ddot{\bm x}_2 \,\, \ldots \bm x_2^{(\mu)}]^\transp\right)\rank(I_{n_2})
        \\
        &= \rank\left([A_{\rm c}^0 \, A_{\rm c}^1 \, \ldots \, A_{\rm c}^{\mu}]\otimes\bm x_1\right)\cdot n_1
        \\
        &\,\,\,\,\,\, + \rank\left(A_{22}^{-1}A_{21}[A_{\rm c}^0 \, A_{\rm c}^1 \, \ldots \, A_{\rm c}^{\mu}]\otimes\bm x_1\right)\cdot n_2
        \\
        &=
        \rank\left([A_{\rm c}^0 \, A_{\rm c}^1 \, \ldots \, A_{\rm c}^{\mu}]\right)\cdot\rank(\bm x_1) \cdot n_1
        \\
        &\,\,\,\,\,\, + \rank\left(A_{22}^{-1}A_{21}[A_{\rm c}^0 \, \ldots \, A_{\rm c}^{\mu}]\right)\cdot \rank(\bm x_1) \cdot n_2
        \\
        &= n_1^2 + n_2^2,
    \end{align*}
    \end{subequations}
    \noindent
    where we have applied in the equalities above the relation $\rank(A\otimes B) = \rank(A) \rank(B)$.
    In the third equality, Assumption \ref{assump.index1lineardae} implies that $\dot{\bm x}_1 = A_{\rm c}\bm x_1$ and $\bm x_2 = A_{22}^{-1}A_{21}\bm x_1$. In the fourth equality, note that $\rank(\bm x_1) =1$ for all $\bm x_1\neq 0$. Moreover, since $A_{22}$ is invertible due to Assumption \ref{assump.index1lineardae} and $\rank(A_{21}) = \min\{n_2,n_1\}$ is satisfied, it follows that 
    \begin{equation*}
        \rank(A_{22}^{-1}A_{21}[A_{\rm c}^0 \,\, A_{\rm c}^1 \,\, \ldots \,\, A_{\rm c}^{n-1}]) = n_2.
    \end{equation*}
    Since $p=n_1^2 + n_2^2$ (size of $\bm\theta$), the identifiability of system \eqref{eq.lineardescriptor.semiexplcit} with respect to $\bm\theta = \operatorname{vec}(A_{11},A_{12})$ follows from Corollary~\ref{cor.idflinearfullmeas}.

    Analogously, for the case of $\bm\theta = \operatorname{vec}(A_{12},A_{21})\in\R^p$, the identifiability of \eqref{eq.lineardescriptor.semiexplcit} follows from Corollary \ref{cor.idflinearfullmeas} given that
    \begin{align*}
        \rank(\mathcal I_{11}) &= \rank\left(\dfrac{\partial \bm\Gamma_{12}}{\partial \operatorname{vec}(A_{12})} \right) + \rank\left(\dfrac{\partial \bm\Gamma_{21}}{\partial \operatorname{vec}(A_{21})} \right)
        \\
        &= \rank\left(A_{22}^{-1}A_{21}[A_{\rm c}^0 \, A_{\rm c}^1 \, \ldots \, A_{\rm c}^{\mu}]\otimes\bm x_1\right)\cdot n_1
        \\
        &\,\,\,\,\,\, + \rank\left([A_{\rm c}^0 \, A_{\rm c}^1 \, \ldots \, A_{\rm c}^{\mu}]\otimes\bm x_1\right)\cdot n_2
        \\
        &= 2n_1n_2 
        = p. \hspace{4.5cm}\QEDS
    \end{align*}

\noindent
\textbf{PROOF OF COROLLARY~\ref{cor.idftlinear.ode}.}
    Since $E = I_n$, it follows that $\bm x^{(\mu)} = A^\mu \bm x$. Thus,
    \begin{equation*} \label{eq.proof.rankI11}
    \begin{aligned}
    \rank\left(
    \mathcal I_{11}(\bm\theta,\bm z) \right)
    &=
    \rank\left( [A^0 \,\, A^1 \,\, \ldots \,\, A^{n-1}]\otimes\bm x\right)\rank(I_n) 
    \\
    &= n^2.
    \end{aligned}
    \end{equation*}
    \noindent
    As in the proof of Corollary \ref{cor.idftlinear.semiexplict}, the second equality follows from the fact that $\rank([A^0 \,\, A^1 \,\, \ldots \,\, A^{n-1}]) = n$ and $\rank(\bm x) =1 $ for $\bm x\neq 0$. Given that $p=n^2$ (size of $\bm\theta$), the identifiability of system \eqref{eq.lineardescriptor.augmented} follows from Corollary \ref{cor.idflinearfullmeas}.  \QEDS

\begin{small}

\end{small}

\end{document}